\documentclass{article}
\usepackage{graphicx,natbib,amsmath,amsfonts,amsthm,todonotes} 
\usepackage{bbm}
\usepackage{xcolor}
\usepackage{enumerate}
\usepackage{hyperref}
\usepackage{comment}

\newcommand{\F}{{\cal F}}

\newcommand{\T}{{\cal T}}

\newcommand{\blue}{\textcolor{black}}
\newcommand{\violet}{\textcolor{black}}

\newcommand{\reals}{\mathbb{R}}

\newtheorem{theorem}{Theorem}
\newtheorem{proposition}{Proposition}
\newtheorem{assumption}{Assumption}

\newtheorem{lemma}{Lemma}

\theoremstyle{definition}
\newtheorem{example}{Example}
\newtheorem{remark}{Remark}
\title{Optimal stopping and divestment timing under scenario ambiguity and learning\footnote{Peter Tankov gratefully acknowledges financial support of the ANR under the program France 2030 (grant ANR-23-EXMA-0011) and from
the FIME Research Initiative of the Europlace Institute of Finance. We thank Nicole B\"auerle and Denis Belomestny  for insightful comments on an earlier draft of this paper. }}

\author{Andrea Mazzon \\ University of Verona \\ \texttt{andrea.mazzon@univr.it}\and Peter Tankov \\ CREST, ENSAE, Institut Polytechnique de Paris \\ \texttt{peter.tankov@ensae.fr}}
\date{}
\begin{document}
\maketitle

\begin{abstract}
    Aiming to analyze the impact of environmental transition on the value of assets and on asset stranding, we study optimal stopping and  divestment timing decisions for an economic agent whose future revenues depend on the realization of a scenario from a given set of possible futures. Since the future scenario is unknown and the probabilities of individual prospective scenarios are ambiguous, we adopt the smooth model of decision making under ambiguity aversion of Klibanoff et al (2005), framing the optimal divestment decision as an optimal stopping problem with learning under ambiguity aversion. We then prove a minimax result reducing this problem to a series of standard optimal stopping problems with learning. The theory is illustrated with two examples: the problem of optimally selling a stock with ambigous drift, and the problem of optimal divestment from a coal-fired power plant under transition scenario ambiguity. 
\end{abstract}


\medskip


\section{Introduction}
Transitioning towards a low-carbon economy is paramount for mitigating the adverse impacts of climate change. The low-carbon transition creates both risks and opportunities for economic agents and financial institutions. Determining the value of assets affected by these risks and opportunities, and making the optimal investment decisions relative to these assets is therefore of key importance. 
The modern approach to this problem involves quantifying possible  future evolutions of the economy and the environment by means of transition scenarios \citep{krey2014global}. 


Produced with integrated assessment models, transition scenarios are disseminated by globally recognized entities, including the IEA (International Energy Agency), IPCC (Intergovernmental Panel on Climate Change), NGFS (Network for Greening the Financial System), and IIASA (International Institute for Applied System Analysis). Companies and financial institutions use these scenarios for evaluating the impact of climate change and energy transition on their portfolios and assets \citep{acharya2023climate}. These stress testing exercises are often carried out by taking a known scenario, and assuming that the bank or company balance sheet remains static throughout the time span of the scenario. However, the time horizon of transition scenarios is usually very long, making the static balance sheet assumption questionable: financial institutions will readjust their portfolios, and companies will make investment or divestment decisions to adapt to the future evolution of the economy and climate. When making these decisions, the economic agents do not know the true scenario, and even the probabilities of realisation of various possible scenarios are ambiguous. Therefore, management of long-term climate related risks and evaluation of transition-sensitive assets must take into account the uncertainty of future transition scenarios.


The main goal of this paper is to understand the impact of scenario uncertainty and \blue{model} ambiguity on investment and divestment decisions of economic agents facing transition-related risks. We distinguish the setting of scenario uncertainty, when several future scenarios are possible, and their probabilities are known to the agent, from that of \blue{model ambiguity on the set of scenarios}, when several future scenarios are possible, but the probabilities of their realization are ambiguous. 

Assume that \violet{the variable which characterizes the realized scenario takes values in a compact metric set $\mathcal{S}$, such that scenario $\theta\in \mathcal S$ corresponds to a probability measure $\mathbb P^\theta$ {on the space of trajectories of the processes observed by the agent}.} A scenario therefore corresponds to a probabilistic model, and \blue{model} 
ambiguity in our setting corresponds to ambiguity over a set of alternative models which describe future evolution of a set of economic variables.

 An economic agent aims to optimally choose a stopping time $\tau$, which may correspond, for example, to the sale or the closure of a polluting asset, or to the investment into a green project. Let $Y_\tau$ be the payoff of the agent at the stopping time. If the probability distribution over the set of future scenarios is known and denoted by $\mathcal P$,
the problem of the agent writes as follows:
$$
\sup_{\tau \in \T}\int_{\mathcal S}d \mathcal P(\theta) \mathbb E^\theta \left[Y_{\tau}\right].
$$
Here, $\mathcal T$ is the set of admissible stopping times. Since the true scenario is not known by the agent, the stopping time may not depend on the scenario. However, the stopping decision may depend on a signal (e.g., observed carbon emissions), which may progressively reveal incomplete scenario information through a learning process. 

In this paper, we are interested in the \blue{case} where the prior scenario \violet{probability distribution is} also unknown. 
We work in the framework of smooth model of decision making of \cite{klibanoff2005smooth}, which in the present setting boils down to solving the following problem: 
\violet{\begin{equation}\label{eq:introduction}
    \sup_{\tau \in \T} v^{-1}\left( \int_{\mathcal S} d\mathcal P(\theta)v\left(\mathbb E^{\theta}\left[Y_{\tau}\right]\right)\right).
    \end{equation}}%
Here, \violet{$\mathcal P$} still plays the role of a reference probability measure on the set of scenarios, but different scenario payoffs are distorted by the agents in their optimization problem through the function $v: \reals \to \reals$, a concave function representing the aversion of investors towards the ambiguity. 
As a result of this distortion, the agents affect more importance to pessimistic scenarios (with low payoffs) even if their reference probability is small. 

Our contribution is two-fold: on the more theoretical side, we study for the first time optimal stopping problems in the setting of \cite{klibanoff2005smooth}, by building on the general framework of \cite{drapeau2013risk}, \blue{which represents preference functionals as penalized expectations over alternative probability measures. In our case, this allows us to rewrite the smooth ambiguity criterion \eqref{eq:introduction} as the sup--inf problem \eqref{eq:generaldualdrapeau}--\eqref{eq:Gfordrapeau}. 
For each stopping time, the agent selects the most adverse scenario distribution over $\mathcal{S}$, while accounting for a flexible form of penalization that increases with the divergence from the reference prior.}
\blue{Our main results establish conditions that guarantee both the interchangeability of the supremum and infimum in this problem, and the existence of a solution.} Being able to swap the order of the supremum and infimum is crucial  because it allows to address the original, non-standard optimal stopping problem by solving a classical, inner optimal stopping problem, using well-known techniques.
\blue{These results are connected to \cite{belomestny2016optimal}, where the authors study optimal stopping problems in terms of conditional convex risk measures and rely on a dual representation involving randomized stopping times in order to prove their minimax result. However, while for the case of discrete time and a finite set of scenarios we rely on the representation of randomized stopping times and the associated topological structure introduced in \cite{belomestny2016optimal}, our analysis builds on the specific features of model ambiguity in our setting. In particular, the structure of scenario uncertainty, combined with suitable continuity assumptions on both the scenario space and the process $Y$, allows us to extend our results to continuous time and uncountable sets of scenarios via a limit argument tailored to the penalization structure induced by smooth ambiguity preferences. Furthermore, unlike \cite{belomestny2016optimal}, where the penalty is incorporated additively with respect to the expected payoff, our formulation accommodates a broader class of penalizations.
}

On the more applied side, we present two examples showing how the present framework can be used to numerically solve two optimal divestment problems with \blue{model ambiguity on the set of scenarios} and learning. 

In the first example, the goal is to find the optimal time to sell a stock, whose expected return value is ambiguous, when the seller acquires scenario information by observing the price of the stock. In the second example, inspired by energy transition, the goal is choose the optimal time to decomission a carbon intensive power plant, and the agent acquires information progressively by observing a signal such as carbon emissions. 
 In both examples, our theoretical results are of great importance  because they permit to significantly simplify the algorithm and exploit dynamic programming. 

\paragraph{Review of literature}
Our paper is part of the extensive literature on optimal stopping problems under model ambiguity. Among the first contributions on the topic, \cite{riedel2009optimal} and \cite{cheng2013optimal} consider the problem
\begin{equation}\label{eq:problemriedel}
\text{maximize $\inf_{\mathbb P \in \mathcal M} \mathbb{E}^{\mathbb P}[X_{\tau}]$ over all stopping times $\tau \le T$}
\end{equation}
for a finite time horizon $T>0$, where $X$ is a stochastic process in  discrete and continuous time, respectively, and $\mathcal M$ is a set of probability measures accounting for model ambiguity.  In \cite{chen2002ambiguity} the problem \eqref{eq:problemriedel} is solved via the use of backward stochastic differential equations (BSDEs).

In this setting, the set $\mathcal M$ has two important properties. First,  all measures in $\mathcal M$ are equivalent to a reference probability measure: a typical example is given by drift uncertainty for the process $X$. Moreover, $\mathcal M$ is constructed in such a way that the optimization problem is \emph{time consistent}, which means that optimal decisions are not changed in time due to the arrival of new pieces of information. In multiple-prior models of type \eqref{eq:problemriedel}, time consistency reduces to the property of stability of the set of priors under pasting: if $\mathbb P_1$ and $\mathbb P_2$ are two elements of $\mathcal M$, then the probability obtained by pasting the density of $\mathbb P_1$ before some stopping time $\tau$ and the density of $\mathbb P_2$ after this stopping time, must also belong to $\mathcal M$. In particular, this implies that the set $\mathcal M$ cannot be finite. 
Time consistency restricts the class of admissible models, but allows on the other hand to apply the principle of dynamic programming. An optimal stopping problem under a sublinear expectation operator accounting for model uncertainty and also satisfying time consistency is studied in \cite{bayraktar2011optimalone} and solved in \cite{bayraktar2011optimaltwo} using BSDEs. Moreover, a continuous-time optimal stopping problem involving dynamic convex risk measures satisfying time consistency is studied in \cite{bayraktar2010optimal}.

A relevant work for our paper is  \cite{belomestny2016optimal}, where optimal stopping problems are studied for conditional convex risk measures involving a worst case expectation with a penalty term. Note that the time consistency property of the risk measures needed in \cite{bayraktar2010optimal} is here dropped.
 Similar tools are also used in \cite{belomestny2017optimal} to show a primal representation result for optimal stopping problems under probability distortions, which have been initially proposed and studied in \cite{xu2013optimal}. 
For further contributions on optimal stopping under non-dominated sets of probability measures we refer among others to  \cite{el1997reflected}, \cite{karatzas2005game},  \cite{karatzas2008martingale}, \cite{follmer2011stochastic}.
Finally, \cite{nutz2015optimal} and \cite{ekren2014optimal} consider the case when the set of probability measures $\mathcal M$ is non-dominated and employ the Snell envelope characterization to study an optimal stopping problem for the supremum of expectations over $\mathcal M$, supposed to be weakly compact, and the existence of optimal stopping times in a zero-sum game between a stopper and a controller choosing a probability measure in $\mathcal M$, respectively. 
In \cite{nutz2015optimal}, the set $\mathcal M$ is supposed to satisfy a \emph{stability under pasting} assumption which is analogous to time consistency in the case of dominated measures, see point (iii) of their Assumption 2.1. A similar condition is called \emph{concatenation property} in \cite{ekren2014optimal}, see their Property (P3).

All contributions listed above take ambiguity aversion into account by optimizing the payoff expectation under the worst case probability measure over a set $\mathcal M$. While the robust worst-case approach may be well suited for the management of catastrophic climate riks \citep{kunreuther2013risk,millner2013welfare,weitzman2009modeling}, it is less easy to justify in the context of cost-benefit evaluation of investments, since agents whose decisions are grounded in the worst-case approach may lose attractive investment opportunities contingent to more favorable scenarios.  For this reason, and unlike the works quoted above, we consider instead the smooth ambiguity setting of \cite{klibanoff2005smooth}, which allows in general for a less conservative approach to ambiguity aversion. Smooth ambiguity adjustments have been shown to be equivalent to robust prior adjustments with a logarithmic specification of risk aversion in \cite{hansen2018aversion}, where the authors construct recursive representations of intertemporal preferences that allow for both penalized and smooth ambiguity aversion to subjective uncertainty. Uncertainty aversion with penalization proportional to discounted relative entropy with respect to ``structured'' models is studied in \cite{hansen2022structured} in a setting where a decision-maker has ambiguity about a
prior over the set of structured statistical models and fears that each of those models is misspecified. \cite{bauerle2023optimal} introduce smooth ambiguity into a portfolio optimization problem. 


In the context of long-term climate-related or transition-related risks, or cost-benefit analysis of climate policies, it is natural to assume that information about scenarios or key parameters of the climate system, such as the climate sensitivity, is not known at the start but discovered progressively by the agents through a learning process \citep{ekholm2018climatic}. This imposes a specific structure on the probability measure $\mathbb P^\theta$ associated to each scenario $\theta$: indeed, under a Bayesian learning framework, the density of each probability measure must evolve according to the Bayesian update rule, see e.g., \cite{flora2023green}. On the other hand, the time consistency property imposes a different structure on the set of probabilities $\mathcal M$ (stability by pasting). Combining the two frameworks (Bayesian learning and time consistency) is not straightforward and leads to very strong constraints on the set $\mathcal M$ (see, e.g., \cite{epstein2022optimal}). Since in this paper we allow progressive learning of scenario information, and {in the examples} we work with a finite set of scenarios, we drop the time consistency assumption. Time-inconsistent optimal stopping problems have been studied in the literature and a variety of approaches have been explored 
\citep{bayraktar2019time,pedersen2016optimal,christensen2018finding,christensen2020time}. 
We determine the optimal stopping strategy at time zero, and assume that the agent precommits to this strategy until the end. \blue{This reflects the idea that the agent evaluates possible scenarios based on a prior belief and penalizes alternative beliefs according to their plausibility, both assessed at the initial time. The stopping rule thus reflects a time-zero valuation that incorporates ambiguity attitudes, which break time consistency.} The alternative, game-theoretical approach, is difficult to implement in continuous time from both the theoretical and the computational point of view, \blue{as it often lacks uniqueness and in general requires strong  assumptions on the model and preferences, \citep{bayraktar2019equilibrium,christensen2020time}.}

From a more applied viewpoint, our paper makes a contribution to the theory of optimal investment under uncertainty. The literature on optimal investment and divestment problems and real options in the context of climate finance and energy project valuation is vast: \cite{detemple2020value} consider the case of a firm that aims to build a new power plant and has to decide between wind and gas technologies, \cite{boomsma2012renewable} adopt a real options approach to analyze investment timing for renewable energy projects under different policy interventions, \cite{abadie2011optimal} assess the problem of finding the optimal time to dismiss a coal station and obtain its salvage value, \cite{laurikka2006emissions} study the  impacts of the European Union Emission allowance Trading Scheme on decisions relative to  investments in a coal power plant, \cite{flora2020price} take the point of view of a  greenhouse gas emitter who can switch from its current high-carbon technology to a cleaner one, and \cite{hach2016capacity} apply a  real option approach to evaluate investment decisions and timing of a single investor in gas-fired power generation. The problem of choosing when to invest in a renewable energy project under uncertainty about future feed-in tariff (FIT)  is studied in \cite{hagspiel2021green}, where it is shown that the range of FITs for which it is optimal
to invest immediately decreases the longer a subsidy has been in place. \cite{dalby2018green} incorporate Bayesian learning into the setting and show that agents are less likely to invest when the likelihood of a policy change increases.

More specifically, uncertainty related to climate transition scenarios is taken into consideration, e.g., in \cite{dumitrescu2024energy}, where the authors develop a theory of optimal stopping mean field games with common noise and partial information to study the impact of  scenario uncertainty on the rate at which conventional generation is replaced by renewable plants in the presence of many interacting agents. In \cite{basei2024uncertainty}, a continuous-time real-options model of green technology adoption subject to scenario uncertainty is developed. In \cite{flora2023green}, an optimal investment/divestment problem is studied in a discrete time setting where the agent faces a set of possible scenarios with known probabilities, obtains scenario information progressively by observing a signal, and updates posterior scenario probabilities through Bayesian learning.  One of the applications of the present paper takes inspiration from \cite{flora2023green} and extends it to the setting when even the initial probability of occurrence of every scenario is unknown to the agent, that is, the agent faces both scenario uncertainty and model ambiguity.


The remainder of the paper is structured as follows. In Section \ref{main}, we define the optimization problems and state the main theoretical results. In Section \ref{sec:example1} we discuss the first example: the problem of optimally selling a stock with ambiguous drift value. Section \ref{sec:example2} is devoted to the second example: optimal closure of a coal-fired power plant under \blue{model ambiguity on the set of transition scenarios}.  Proofs of the main theoretical results are gathered in the Appendix. 
 
 \section{Setting and main results}
\label{main}

\subsection{Setting and main assumptions}
Let $0<T<\infty$ be a time horizon and let $\left(\Omega, {\mathcal{F}},\mathbb F = \left({\mathcal{F}}_t\right)_{0 \leq t \leq T}, \mathbb{P}\right)$ be a filtered probability space satisfying the usual conditions. 
Let $\T$ be the set of all $\mathbb F$-stopping times $\tau$ on $[0,T]$.
Introduce an $\mathbb F$-adapted, right-continuous stochastic process $Y=(Y_t)_{0 \le t \le T}$, which corresponds to the process of risk factors observed by the agent.

Let $\mathcal S$ be a compact metric space indexing the possible scenarios (for example, a discrete set of $N$ possible scenarios together with the discrete metric). 
We denote by $\mathcal P$ a reference probability distribution on $\mathcal S$.\footnote{For the sake of clarity, we adopt the convention that measures associated to trajectories are denoted using \emph{blackboard bold} notation (e.g., $\mathbb{P}$), while measures on scenarios are denoted using \emph{calligraphic} notation (e.g., $\mathcal P$, $\mathcal Q$).}

We assume that for any scenario $\theta\in \mathcal S$, there exists a probability measure on $(\Omega,\mathcal F)$, denoted by $\mathbb P^\theta$, with the associated expectation operator denoted by $\mathbb{E}^\theta[\cdot]$. The following assumptions will be used throughout the paper. 
\begin{assumption}\label{ass:equivalence}
All probability measures $\mathbb P^{\theta}$, $\theta \in \mathcal S$, are absolutely continuous with respect to  $\mathbb P$.
\end{assumption}

\begin{assumption}\label{ass:integrability}
There exists a constant $C<\infty$ such that 
\begin{equation}\label{eq:integrabilityconditiononf}
\mathbb{E}^{\theta}\left[\sup_{0 \le t \le T} \left|Y_t\right|\right]<C\quad \forall \theta\in \mathcal S.
\end{equation}
\end{assumption}

\begin{assumption}\label{holder.ass}
For any sequence $(\tau_n) \in \T$ and any $\tau \in \T$ with $\tau_n \to \tau$ a.s. it holds
$$
\mathbb E^{\theta}[Y_{\tau_n} - Y_\tau] \to 0,
$$
uniformly on $\theta \in \mathcal S$.
\end{assumption}

The following technical assumption holds trivially when the set of scenarios $\mathcal S$ is finite. { We denote by $D^\theta$ the Radon-Nikodym derivative of $\mathbb P^\theta$ with respect to $\mathbb P$ and by $d(\cdot,\cdot)$ the distance on the metric space $\mathcal S$. If $\mathcal S$ is finite, we set $d(\theta_1,\theta_2) =0$ when $\theta_1 = \theta_2$ and $1$ otherwise.} 
\begin{assumption}\label{cont.ass}
$$
\sup_{\theta_1,\theta_2 \in \mathcal S: d(\theta_1,\theta_2)\leq h} \mathbb E[|D^{\theta_1}-D^{\theta_2}| \sup_{0\leq t\leq T} |Y_t|]\to 0
$$
as $h\to 0$. 
\end{assumption}

{Below, we provide two examples of settings where our assumptions 1-4 are satisfied. We start with an abstract general example. 
\begin{example}\label{gen.ex}
Let $\mathcal S$ be a compact subset of $L^2(\mathbb P)$, such that for every $D\in \mathcal S$, $D\geq 0$ a.s., and $\mathbb E[D]=1$.  For each $D\in \mathcal S$, the measure $\mathbb P^D$ is defined through its Radon-Nikodym derivative 
$$
\frac{d\mathbb P^D}{d\mathbb P} = D. 
$$
Let $Y$ be an Itô process with jumps of the form
$$
Y_t = Y_0 + \int_0^t \gamma_s ds + \int_0^t \sigma_s dW_s + \int_{0}^t \int_{\mathbb R} \zeta_s(x) \tilde \mu(ds,dx),
$$
where $\mathbb E[Y_0^2]<\infty$, $W$ is a Brownian motion, $\tilde\mu$ is a compensated Poisson random measure with intensity measure $ds \times \nu$, where $\nu$ is a Lévy measure, and $\gamma$, $\sigma$ and $\zeta$ are adapted random coefficients such that 
\begin{align}
\mathbb E\left[\int_0^t (\gamma^2_s + \sigma^2_s + \int_{\mathbb R}\zeta_s^2(x)\nu(dx) )ds\right] <\infty.\label{int.ass}
\end{align}
Assumption 1 is trivially satisfied. For Assumption 2, we write
$$
\mathbb E^{D}[\sup_{0\leq t\leq T} Y_t] \leq \mathbb E \left[D^2 \right]^{\frac{1}{2}}\mathbb E \left[ \sup_{0\leq t\leq T} Y_t^2  \right]^{\frac{1}{2}}.
$$
The first factor is uniformly bounded by compactness, and the second factor is finite by Assumption \ref{int.ass} (using Burkholder inequality for the martingale terms). For Assumption 3 we similarly write: 
$$
\big|\mathbb E^{D}[Y_{\tau_n} - Y_\tau]\big|\leq \mathbb E \left[D^2 \right]^{\frac{1}{2}}\mathbb E \left[(Y_{\tau_n} - Y_\tau )^2\right]^{\frac{1}{2}},
$$
and 
$$
\mathbb E \left[(Y_{\tau_n} - Y_\tau )^2\right] \leq 3 \mathbb E \left[\left(\int_{\tau}^{\tau_n}\gamma_t dt\right)^2\right] + 3 \mathbb E \left[\int_{\tau}^{\tau_n}\sigma^2_t dt\right]+3 \mathbb E \left[\int_{\tau}^{\tau_n}\int_{\mathbb R}\zeta^2_t(x)\nu(dx) dt\right],
$$
where all terms converge to zero by the dominated convergence theorem. 
Finally, for Assumption 4 we have:
\begin{multline*}
\sup_{D_1,D_2 \in \mathcal S: \|D_1-D_2\|_2\leq h} \mathbb E[|D_1-D_2| \sup_{0\leq t\leq T} |Y_t|] \\ \leq \sup_{D_1,D_2 \in \mathcal S: \|D_1-D_2\|_2\leq h} \mathbb E[|D_1-D_2|^2]^{\frac{1}{2}} \mathbb E[\sup_{0\leq t\leq T} |Y_t|^2]^{\frac{1}{2}}\leq h\mathbb E[\sup_{0\leq t\leq T} |Y_t|^2]^{\frac{1}{2}},
\end{multline*}
which clearly tends to zero as $h\to 0$. 
\end{example}

As we saw in the previous example, beyond the standard integrability assumptions, one condition that may be difficult to satisfy is the norm compactness of the set of measure changes corresponding to different scenarios in the space $L^2(\mathbb P)$. When the set $\mathcal S$ is finite-dimensional, such as when the different scenarios represent parameter uncertainty in a parametric model, compactness is generally easy. When the set $\mathcal S$ is infinite-dimensional, to obtain compactness one needs to impose more structure and regularity on the set of measure changes and use, for example, Sobolev-style embeddings. We discuss such an example below.

\begin{example}[Brownian motion with deterministic drift]
Let 
$$
Y_t = y + \sigma W_t,
$$
where $W$ is a Brownian motion, and let $\mathcal S$ be a subset of functions  from $[0,T]$ to $\mathbb R$, which is bounded in $H^1$, that is,
$$
\sup_{\varphi \in \mathcal S} \left\{\int_0^T  \varphi(t)^2 dt+ \int_0^T \dot \varphi(t)^2 dt\right\}<\infty 
$$
For each $\varphi \in \mathcal S$, the measure $\mathbb P^\varphi$ is defined through its Radon-Nikodym derivative
$$
D^\varphi = e^{\int_0^T \varphi_t d W_t - \frac{1}{2}\int_0^T\varphi_t^2 dt }. 
$$
Then, under $\mathbb P^\varphi$, $W^\varphi_t = W_t - \int_0^t \varphi_s ds$ is a Brownian motion, and the process $Y$ follows the dynamics
$$
Y_t  =y + \int_0^t \sigma\varphi_s ds + \sigma W^\varphi_t.
$$
The set $\mathcal S$ is compact in $L^2([0,T])$ by Rellich-Kondrachov theorem \cite[Theorem 6.3]{adams2003sobolev}, and the map $\varphi \mapsto D^\varphi$ is a continuous map from $\mathcal S$ to $L^2(\mathbb P)$, because
\begin{align*}
&\mathbb E[(D^{\varphi_1}- D^{\varphi_2})^2] = e^{\int_0^T\varphi_1(t)^2 dt}+e^{\int_0^T\varphi_2(t)^2 dt} - 2 e^{\int_0^T\varphi_1(t)\varphi_2(t) dt}\\
& = \left(e^{\frac{1}{2}\int_0^T\varphi_1(t)^2 dt}-e^{\frac{1}{2}\int_0^T\varphi_2(t)^2 dt}\right)^2 \\ &\qquad+ 2 e^{\frac{1}{2}\int_0^T\varphi_1(t)^2 dt+\frac{1}{2}\int_0^T\varphi_2(t)^2 dt}(1-e^{-\frac{1}{2}\int_0^T (\varphi_1(t)-\varphi_2(t))^2 dt})\\
&\leq e^{\max_{\varphi\in \mathcal S} \int_0^T \phi(t)^2 dt}\left\{\frac{1}{4} \left(\int_0^T |\phi_1(t)^2-\phi_2(t)^2| dt\right)^{{2}}+ \int_0^T (\phi_1(t)-\phi_2(t))^2 dt\right\}.
\end{align*}
Thus, assumptions of Example \ref{gen.ex} are satisfied. To define a probability measure on $\mathcal S$, consider the following orthonormal basis of $H^1([0,T])$: 
\begin{align*}
\phi_k(t) &= \frac{e_k(t)}{\sqrt{1+\lambda_k}},\quad k=0,1,2,\dots\\
e_0(t) &= \frac{1}{\sqrt{T}},\quad e_k(t) = \sqrt{\frac{2}{T}}\cos\left(\frac{2\pi k t}{T}\right),\quad k=1,2,\dots\\
\lambda_k &= \left(\frac{k\pi}{T}\right)^2,
\end{align*}
let $(u_k)_{k\geq 0}$ be a sequence of i.i.d.~random variables uniformly distributed on $[-1,1]$, and $(\alpha_k)$ a sequence of constants with $\sum_{k\geq 0}\alpha_k^2 <\infty$, and define
$$
X_t = \sum_{k\geq 0} u_k \alpha_k \phi_k(t).
$$
Then,
$$
\|X\|^2_{H^1} =\sum_{k\geq 0} u_k^2 \alpha_k^2 \leq \sum_{k\geq 0} \alpha_k^2 <\infty. 
$$
Thus, the support of the distribution of $X$ is a compact subset of $L^2([0,T])$. 
\end{example}

}
\subsection{Problem formulation}
 In the absence of model ambiguity concerns, the agent is interested in the optimal stopping problem
$$
 \sup_{\tau \in \T}\int_{\mathcal S} d \mathcal P(\theta)\mathbb E^\theta \left[Y_{\tau}\right].
$$
Following the smooth model of decision making under ambiguity of \cite{klibanoff2005smooth}, we assume that  the agent solves instead the optimization problem
\begin{equation}\label{eq:mainproblem}
\sup_{\tau \in \T} v^{-1}\left( \int_{\mathcal S} d\mathcal P(\theta)v\left(\mathbb E^{\theta}\left[Y_{\tau}\right]\right)\right),
\end{equation}
where $v$ is a suitable continuous function which characterizes the ambiguity concern: the agent is \emph{ambiguity loving} if $v$ is convex and \emph{ambiguity averse} if $v$ is concave.

Hereafter, we call $\mathcal M(\mathcal S)$ the set of probability measures on $\mathcal S$, endowed with the topology of weak convergence. Moreover, for any $\mathcal Q \in \mathcal M(\mathcal S)$ we denote by $\mathbb P^{\mathcal Q}$ the probability measure on $\mathcal B(S)\times \mathcal F$ defined by
\begin{equation}\label{eq:PQ}
\mathbb P^{\mathcal Q}(B) = \int_{\mathcal S\times \Omega} d\mathcal Q(\theta) d\mathbb P^\theta(\omega) \mathbf 1_{B}(\theta, \omega). 
\end{equation}
and by $\mathbb E^{\mathbb P^{\mathcal Q}}$
the associated expectation operator.


We focus on ambiguity aversion. In this case, Theorem 6 in \cite{drapeau2013risk} implies the following result. 
\begin{proposition}\label{prop:drapeau}
For any proper concave nondecreasing upper semicontinuous function $v:\mathbb{R}\to \mathbb{R}$, problem \eqref{eq:mainproblem}  admits the representation
\begin{equation}\label{eq:generaldualdrapeau}
\sup_{\tau \in \T} \inf_{\mathcal Q \in \mathcal M(\mathcal S)} G\left(\tau,\mathcal Q\right),
\end{equation}
where 
\begin{equation}\label{eq:Gfordrapeau}
G\left(\tau,\mathcal Q \right) = R\left(\mathcal Q , \mathbb E^{\mathbb P^{\mathcal Q}}[Y_\tau]\right),
\end{equation} 
for a unique function $R: \mathcal M(\mathcal S) \times \mathbb{R} \to \mathbb{R}$ such that:
\begin{enumerate}[(i)]
\item $R(\mathcal Q ,\cdot)$ is nondecreasing and right continuous for any $\mathcal Q \in \mathcal M(\mathcal S)$;
\item $R$ is jointly quasi-convex;

\item $\lim_{s\to+\infty}R(\mathcal Q ^1, s) = \lim_{s\to+\infty}R(\mathcal Q ^2, s)$ for any $\mathcal Q ^1, \mathcal Q ^2 \in \mathcal{M}(\mathcal S)$;
\item the function $R^+(\mathcal Q ,s):=\sup_{s'<s}R(\mathcal Q ,s')$ is lower semicontinuous in the first argument.
\end{enumerate}
\end{proposition}

\begin{remark}
The key insight from Proposition \ref{prop:drapeau} is that the original problem in \eqref{eq:mainproblem} can be expressed in terms of a penalized worst-case scenario approach. Specifically, the investor aims to find the optimal stopping time focusing on the most unfavourable choice of probability measure $\mathcal Q$, as captured by \eqref{eq:generaldualdrapeau}. The choice of probability measure is evaluated through the mapping $G$, given in \eqref{eq:Gfordrapeau}, which depends both on the probability measure $\mathcal Q$ and the expected gain upon stopping in $\tau$, i.e.,
$$\mathbb E^{\mathbb P^{\mathcal Q}}[Y_{\tau}]=\int_{\mathcal S} d\mathcal Q(\theta)\mathbb E ^{\theta}[Y_{\tau}].
$$
The mapping $G$ may  penalize the deviation of the measure $\mathcal Q$ from the reference probability measure $\mathcal P$ in an additive or multiplicative way, see Example \ref{ex:commonfunctions}. 
\end{remark}

{To be able to prove the minimax result, we shall impose the following additional alternative assumptions on the function $R$. They hold in most practical cases and are easy to verify in specific examples (see below).}
\begin{assumption}\label{continuity}
One of the following holds true:
\begin{itemize}
\item[(i)] The function $R$ is continuous in the second argument on $\mathbb R$ for all $\mathcal Q \in \mathcal M(\mathcal S)$;
\item[(ii)] The function $R$ is continuous in the second argument on $(0,\infty)$ for all $\mathcal Q \in \mathcal M(\mathcal S)$ and 
$$
\mathbb E^{\theta}[Y_\tau]>0,\quad \forall \tau \in \mathcal T, \quad \forall \theta \in \mathcal S.
$$
\end{itemize}
\end{assumption}


\begin{example}\label{ex:commonfunctions}
Among the most common choices for the function $v$ we mention the following:
\begin{itemize}
\item \textbf{Power function}:
\begin{align}
v(x) = \begin{cases}x^\lambda \quad &\text{for $x \ge 0$},\\
-\infty \quad &\text{for $x<0$} \end{cases}
\label{eq:powerlambdapositive}
\end{align}
for $\lambda \in (0,1)$ and 
\begin{align}
v(x) = \begin{cases}-x^\lambda \quad &\text{for $x \ge 0$},\\
-\infty \quad &\text{for $x<0$} \end{cases} \label{eq:powerlambdanegative}
\end{align}
for $\lambda <0$. Here the parameter $\lambda < 1$ characterizes the agent's ambiguity aversion: the case when $\lambda \to -\infty$ corresponds to full ambiguity aversion. 

We have dual representation \eqref{eq:generaldualdrapeau} with function $R$ in \eqref{eq:Gfordrapeau} given by
\begin{align}
R(\mathcal Q ,s):=s\mathbf 1_{s>0}\, \mathbb E^{\mathcal P}\left[\left(\frac{d\mathcal Q }{d\mathcal P}\right)^{\frac{\lambda}{\lambda-1}}\right]^{{\frac{1-\lambda}{\lambda}}},
\label{Rpower}
\end{align}
{with the convention that whenever the second factor is infinite, $R(\mathcal Q ,s)=+\infty$ for $s\geq 0$ and $R(\mathcal Q ,s) = 0$ for $s< 0$.} {In this case, for $\lambda < 0$, this function is continuous, so that Assumption \ref{continuity} (i) holds. On the other hand, for $\lambda\in (0,1)$, for $\mathcal Q $ such that the second factor is infinite, $R(\mathcal Q ,s)$ is continuous in $s$, as a function taking values in the extended real line, only on $(0,\infty)$, so that we are in the context of Assumption \ref{continuity} (ii). }

This means that  
\begin{equation}\label{eq:Gpower}
G(\tau,\mathcal Q ):=\left(\mathbb E^{\mathbb P^{\mathcal Q}}[Y_\tau]\right)^+\mathbb E^{\mathcal P}\left[\left(\frac{d\mathcal Q }{d\mathcal P}\right)^{\frac{\lambda}{\lambda-1}}\right]^{{\frac{1-\lambda}{\lambda}}},
\end{equation}
for $\lambda \in (-\infty, 0) \cup (0,1)$, with the same convention as above.


When $\lambda \to -\infty$, the second term on the right-hand side of \eqref{eq:Gpower} is equal to $1$, so that
$$
G(\tau,\mathcal Q ) = \left(\mathbb E^{\mathbb P^{\mathcal Q}}[Y_\tau] \right)^+.
$$
This is the full ambiguity aversion setting, where the agent wants to maximize the expectation under the worst-case prior measure $\mathcal Q $ without caring how much it diverges from $\mathcal P$.

On the other hand, when $\lambda \to 1$, the second term on the right-hand side of \eqref{eq:Gpower} gives $+\infty$ for any $\mathcal Q \ne \mathcal P$, so the agent only considers the original probability measure.

\item \textbf{Logarithmic function}: we define 
\begin{align}\label{logfunction}
v(x) = \begin{cases}\log(x) \quad &\text{for $x \ge 0$},\\
-\infty \quad &\text{for $x\le 0$}. \end{cases}
\end{align}

We have dual representation \eqref{eq:generaldualdrapeau} with function $R$ in \eqref{eq:Gfordrapeau} given by
$$
R(\mathcal Q ,s):=s\mathbf 1_{s>0} \exp\left(-
\mathbb E^{\mathcal P}\left[
\log\left(\frac{d \mathcal Q }{d \mathcal P}\right)\right]\right),
$$
that is,
$$
G(\tau,\mathcal Q ):=\left(\mathbb E^{\mathbb P^{\mathcal Q}}[Y_\tau]\right)^+ \exp\left(-
\mathbb E^{\mathcal P}\left[
\log\left(\frac{d \mathcal Q }{d \mathcal P}\right)\right]\right).
$$
The logarithmic function can also be obtained as the limit of the power function when $\lambda \to 0$. {Here, similarly to the power function, we are in the context of Assumption \ref{continuity} (ii). }

\item \textbf{Exponential function}: $v(x)=-e^{-\gamma x}$, for $\gamma >0$. We have the standard representation for the entropic risk measure with function $R$ in \eqref{eq:Gfordrapeau} of the following form:
$$
R(\mathcal Q ,s):=s- \frac{1}{\gamma}\mathbb E^{\mathcal Q}\left[
\log\left(\frac{d \mathcal Q }{d \mathcal P}\right)\right],
$$
that is,
\begin{equation}\label{eq:Gfunctionexp}
G(\tau,\mathcal Q ):=\mathbb E^{\mathbb P^{\mathcal Q}}[Y_{\tau}]-\frac{1}{\gamma}\mathbb E^{\mathcal Q}\left[
\log\left(\frac{d \mathcal Q }{d \mathcal P}\right)\right].
\end{equation}
Here we are in the context of Assumption \ref{continuity} (i). 
\end{itemize}
\end{example}

\subsection{Main results}
\blue{The formulation in problem \eqref{eq:generaldualdrapeau} allows to get rid of the non-linearity of  \eqref{eq:mainproblem}, but the problem is still untractable since the optimal stopping rule involves the distribution of the payoff over a set of probability measures, so that the worst case is potentially different for every stopping time. For this reason, we} now would like to prove that one can exchange the supremum and the infimum in \eqref{eq:generaldualdrapeau}, \blue{as this would allow to separate in a convenient way the original problem in an inner optimal stopping for any fixed measure and an outer optimization over the set of measures.}

First, we focus on discrete time optimal stopping problems and a finite set of scenarios. In this case, the existence of a saddle point may be established. The following theorem, whose proof is given in Appendix \ref{sec:proofexistence}, is the first main theoretical result of the paper. It is based on the compactness of the set of randomized stopping times, and a mapping 
\blue{from randomized to ordinary stopping times preserving the expected payoff under every scenario.}

\begin{theorem}\label{thm:existencesolution}
Let $v$ satisfy the assumptions of Proposition \ref{prop:drapeau}, let  Assumptions \ref{ass:equivalence}, \ref{ass:integrability} and \ref{continuity} be satisfied and suppose that the set of scenarios $\mathcal S$ is finite. Let $\mathbb{T}:=\left\{t_1, \ldots, t_{r}\right\}$ with $0<t_1<\cdots<t_{r}=T$.  Then there exists a saddle point: $\tau^* \in \mathcal{T}_{\mathbb{T}}$ and $\mathcal Q ^* \in \mathcal{M}(\mathcal{S})$ satisfying
$$
G(\tau,\mathcal Q ^*) \leq G(\tau^*,\mathcal Q ^*) \leq G(\tau^*,\mathcal Q )
$$
for any $\mathcal Q \in \mathcal M(\mathcal S)$ and $\tau \in \mathcal{T}_{\mathbb{T}}$, the set of stopping times with values in $\mathbb T$. In particular,
\begin{equation}\label{eq:minimaxthm1}
 G(\tau^*,\mathcal Q ^*) =\max _{\tau \in \mathcal{T}_{\mathbb{T}}} \min _{\mathcal Q \in \mathcal M(\mathcal S)} G(\tau,\mathcal Q ) =\min _{\mathcal Q \in \mathcal M(\mathcal S)} \max _{\tau \in \mathcal{T}_{\mathbb{T}}} G(\tau,\mathcal Q ).
\end{equation}
\end{theorem}

{In our second minimax result, this theorem is extended to a general set of scenarios. In this case, the one-to-one 
correspondence between randomized and ordinary stopping times breaks down, and the existence of a saddle point is not guaranteed. The proof is also given in Appendix \ref{sec:proofexistence}.} 

\begin{theorem}\label{thm:infscenarios}
{Let $v$ satisfy the assumptions of Proposition \ref{prop:drapeau} let  Assumptions \ref{ass:equivalence}, \ref{ass:integrability}, \ref{cont.ass} and \ref{continuity} be satisfied. Let $\mathbb{T}:=\left\{t_1, \ldots, t_{r}\right\}$ with $0<t_1<\cdots<t_{r}=T$.  Then 
\begin{equation}\label{eq:minimaxthm2}
\sup_{\tau \in \mathcal{T}_{\mathbb{T}}} \inf_{\mathcal Q \in \mathcal M(\mathcal S)} G(\tau,\mathcal Q ) =\inf_{\mathcal Q \in \mathcal M(\mathcal S)} \sup _{\tau \in \mathcal{T}_{\mathbb{T}}} G(\tau,\mathcal Q ).
\end{equation}}
\end{theorem}

We finally extend our minimax result to arbitrary stopping times. 
The proof of the following theorem is provided in Appendix \ref{sec:proofminimax}.
\begin{theorem}\label{thm:minmax}
Let $v$ satisfy the assumptions of Proposition \ref{prop:drapeau} and let Assumptions \ref{ass:equivalence}, \ref{ass:integrability}, \ref{holder.ass}, \ref{cont.ass} and \ref{continuity} be satisfied. 

Then,
\begin{equation}\label{eq:minimaxthm3}
\sup_{\tau \in \T} \inf_{\mathcal Q \in \mathcal M(\mathcal S)} G\left(\tau,\mathcal Q \right)=  \inf_{\mathcal Q \in \mathcal M(\mathcal S)}\sup_{\tau \in \T} G\left(\tau,\mathcal Q\right).
\end{equation}
\end{theorem}

\begin{remark}
The results presented in Theorems \ref{thm:existencesolution}, \ref{thm:infscenarios} and \ref{thm:minmax} allow for a fundamental transformation of the original, non-standard optimal stopping problem under ambiguity aversion. Indeed, the inner problem on the right-hand side of equations \eqref{eq:minimaxthm1}-\eqref{eq:minimaxthm3} is a classical optimal stopping problem under the measure $\mathbb P^\mathcal Q$ defined in \eqref{eq:PQ}, and it can be therefore solved using standard methods \blue{through the dynamic programming principle and backward induction, see for example the numerical applications in Sections \ref{sec:example1} and \ref{sec:example2}}. 
\end{remark}

\blue{
\begin{remark}
The proofs of Theorems  \ref{thm:infscenarios} and \ref{thm:minmax} given in Appendices \ref{sec:proofexistence} and \ref{sec:proofminimax}, respectively, also provide the construction of a nearly-optimal stopping time taking values on a discrete set of times and based on a finite set of scenarios. In particular, it can be seen that for any $\delta>0$, there exist $\bar{m}, \bar{n} \in \mathbb{N}$ such that for all $m \ge \bar{m}$, $n \ge \bar{n}$ there exists a stopping time $\tau^{n,m}$ taking values in  
$\mathbb{T}_m:=\{kT/m,k=1,\dots,m\}$ and selected according to the expectations correseponding to a subset $\mathcal{S}_n \subseteq \mathcal{S}$ of $n$ elements, such that
$$
\inf_{\mathcal{Q} \in \mathcal{M}(\mathcal{S})} G\left(\tau^{n,m},\mathcal{Q}\right) \ge \sup_{\tau \in \T} \inf_{\mathcal Q \in \mathcal M(\mathcal S)} G\left(\tau,\mathcal Q \right) - \delta=  \inf_{\mathcal Q \in \mathcal M(\mathcal S)}\sup_{\tau \in \T} G\left(\tau,\mathcal Q\right) - \delta.
$$
\end{remark}
}

\section{Selling a stock with ambiguous drift}
\label{sec:example1}
In this section, we consider the problem of  optimally selling a stock, when its expected return rate is not known by the investor \blue{and can take values in a compact set of scenarios. We consider both the cases of an ambiguity neutral and an ambiguity averse investor in order to investigate how ambiguity aversion impacts the valuation of this option.} We assume that the stock price follows the Black-Scholes model with unobservable drift:
$$
\frac{dS_t}{S_t} = b^\theta dt + \sigma dW_t,
$$
where $\theta$ is the scenario variable. Letting $X_t = \log S_t$, we have:
$$
dX_t = \mu^\theta dt + \sigma dW_t,\quad \mu^\theta = b^\theta -\frac{\sigma^2}{2}.
$$

As in Section \ref{main}, we denote by $\mathcal P$ the reference probability measure on $\mathcal S$,  by $\mathbb P^{\mathcal P}$ the associated probability measure on $\mathcal B(S)\times \mathcal F$ defined according to
\eqref{eq:PQ} \blue{and by $\Theta$ the scenario variable under $\mathbb P^{\mathcal P}$}. 

The following proposition (see e.g., \cite{bismuth2019portfolio}) provides the value of $\beta_t:=\mathbb E^{\mathbb P^{\mathcal P}}[\mu^{\Theta}|\mathcal F_t]$ and characterizes the dynamics of the price process under the posterior probability. 
\begin{proposition}\label{drift.uncert.prop}
Assume that the prior distribution $m_\mu$ of $\mu^{\Theta}$ is sub-Gaussian, that is, there exists $\eta>0$ with
\begin{equation}\label{eq:subgaussian}
\int e^{\eta z^2} m_\mu (dz)<\infty.
\end{equation}
Then, the following holds true:
\begin{itemize}
\item For all $t\geq 0$,
$$
\beta_t =  \Gamma_m(t,X_t),
$$
where 
\begin{equation}\label{eq:generalGm}
\Gamma_m(t,x) = \frac{\int m_\mu(dz) z\exp\left(-\frac{t}{2}\frac{z^2}{\sigma^2} + \frac{z(x-X_0 )}{\sigma^2}\right)}{\int m_\mu(dz) \exp\left(-\frac{t}{2}\frac{z^2}{\sigma^2} + \frac{z(x-X_0 )}{\sigma^2}\right)}.
\end{equation}
\item The process $\widehat W$ defined by 
$$
\widehat W_t = W_t + \int_0^t \sigma^{-1}(\mu^\Theta - \beta_t) ds
$$
is a standard Brownian motion adapted to $\blue{\mathbb{F}}$. 
\item The signal follows the Markovian dynamics
\begin{equation}\label{eq:signalunderposterior}
dX_t =  \Gamma_m(t,X_t) dt +\sigma d\widehat W_t.
\end{equation}
\end{itemize}
\end{proposition}
{ For example, if the distribution $m_\mu$ is supported by two points $\mu_1 = \mu$ and $\mu_2 = -\mu_1$, then $\Gamma_m$ does not depend on $t$ and is given by the simple formula (where $p = m(\mu)$).
$$
\Gamma_m(t,x) = \mu\frac{p \exp\left(\frac{2\mu(x-X_0 )}{\sigma^2}\right)-(1-p)}{p \exp\left(\frac{2\mu(x-X_0 )}{\sigma^2}\right)+(1-p) }
$$}

\blue{Note that condition \eqref{eq:subgaussian} is trivially satisfied by any distribution with compact support, as in our case. In this model, for fixed $X_0 = \log S_0$, there is a one-to-one correspondence between the log stock price value $X_t = \log S_t$ and the best estimate of the drift of the process. The drift in \eqref{eq:signalunderposterior} exhibits time- and path-dependent dynamics through $\Gamma_m(t, X_t)$, which encodes the agent's posterior belief over the unknown scenario parameter. The process $\widehat W$ is called the \emph{innovation process} in filtering theory, and captures the remaining randomness after filtering out the learned component.} 

\blue{Figure \ref{traj.fig} illustrates the learning process in the model. The left graph plots two stock price trajectories with two different volatility values ($\sigma=0.1$ for the low volatility and $\sigma=0.3$ for the high volatility) and drift values, selected randomly among the three possible scenarios. The true drift values are shown on the graph. The right graph shows the trajectories of the posterior estimate of the drift $\mathbb E^{\mathbb P^\mathcal P}[b^\Theta|\mathcal F_t]$, with the three possible values shown as dotted lines. When the volatility and thus the uncertainty is low, the drift quickly converges to the true value, while for higher volatility the oscillations around the true value are more significant.}

\begin{figure}
\centerline{\includegraphics[width=1.2\textwidth]{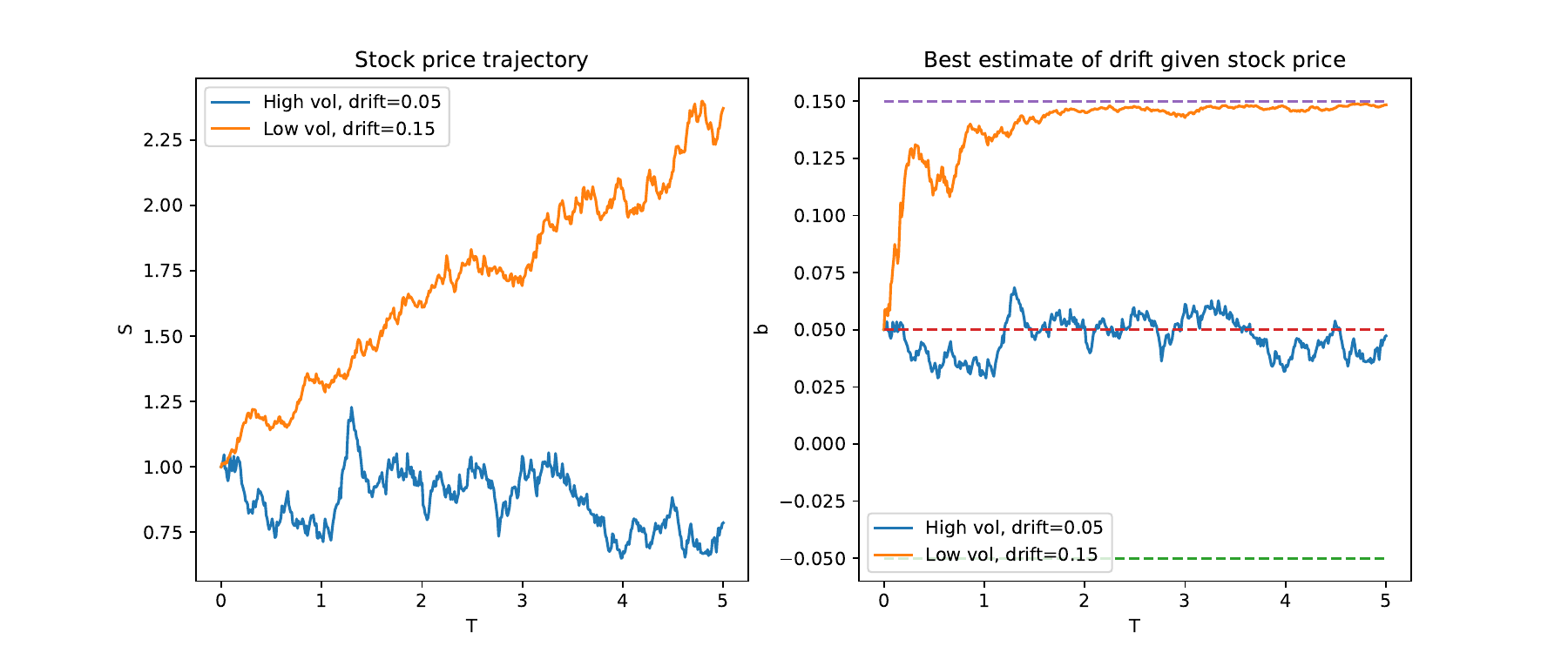}}
\caption{Learning process in the model. The left graph illustrates two stock price trajectories with two different volatility values ($\sigma=0.1$ for the low volatility and $\sigma=0.3$ for the high volatility) and drift values, selected randomly among the three possible scenarios, and shown on the graph. The right graph shows the trajectories of the posterior estimate of the drift.  }
\label{traj.fig}
\end{figure}

\paragraph{Ambiguity-neutral investor}In the absence of ambiguity aversion, the investor wishes to maximize the discounted expected value of the stock:
\begin{align}
\blue{\sup_{\tau \in \mathcal T}} \mathbb E^{\mathbb P^{\mathcal P}}[e^{-r \tau} S_\tau] = \sup_{\tau \in \mathcal T} \mathbb E^{\mathbb P^{\mathcal P}}[ e^{-r \tau + X_\tau}],\label{value.eq}
\end{align}
\blue{where $\mathcal T$ is the set of stopping times with values in $[0,T]$.}
Note that by modifying the drift distribution, we can and will assume without loss of generality that $r=0$. 

Introduce the value function
\begin{align}
v(t,x) = \blue{\sup_{t\leq \tau \leq T}}\mathbb E^{\mathbb P^{\mathcal P}}\left[e^{X_\tau}\Big|X_t = x\right].\label{vf.eq}
\end{align}
By standard arguments \citep{peskir2006optimal}, it satisfies the variational inequality
\begin{align}
\max\left\{\frac{\partial v}{\partial t} + \frac{\sigma^2}{2} \frac{\partial^2 v}{\partial x^2}  +  \Gamma(t,x)\frac{\partial v}{\partial x} ,{e^x-v}\right\}= 0,\label{vareq}
\end{align}
with the terminal condition $v(T,x) = e^x$,
{and the optimal stopping time exists and is given by 
$$
\tau^* = \inf\{t\in [0,T]: v(t,X_t)=e^{X_t}\}. 
$$
The following proposition characterizes this stopping time in terms of an exercise boundary. Remark that when $\mu_1+\sigma^2/2 \geq 0$, it is easy to see that the optimal stopping time is $\tau^* = T$ and when $\mu_2+\sigma^2/2 \leq 0$, then $\tau^* = 0$, so that the proposition considers the nontrivial case. Its proof is given in Appendix \ref{boundary.proof}.
\begin{proposition}\label{boundary.prop}
Let the prior distribution $m_\mu$ have bounded support with endpoints $\mu_1$ and $\mu_2$. Assume that $\mu_1 + \sigma^2/2< 0<\mu_2 + \sigma^2/2$.  
Then there exists a bounded continuous map $b:[0,T]\to \mathbb R$ such that 
$$
\tau^* = \inf\{t\in[0,T]: X_t \leq b(t)\}. 
$$
\end{proposition}}


{
In the first numerical experiment, we analyze the convergence of the optimal stopping boundary and the value function for a discrete set of stopping times to the continuous-time solution described above as the number of possible stopping dates increases. To solve the variational inequality \eqref{vareq}, we use the standard implicit discretization scheme with the optimal stopping rule applied at every step in the continuous-stopping case, or only at the allowed stopping dates. 

Figure \ref{noambig.fig} shows the convergence of the value \eqref{value.eq} and of the optimal stopping boundary in a model without model ambiguity, as the number of possible stopping dates increases. For comparison, in the bottom panel, we also show the value obtained when the scenario uncertainty is immediately resolved.\footnote{In this case, the optimal stopping rule becomes trivial: if the drift value is greater than the interest rate, it is optimal to stop at the terminal date, otherwise immediate stopping is optimal. In particular, this value no longer depends on the volatility $\sigma$. } 
The following parameter values were used: $T = 5$ years, $r = 2\%$,  $b =[-5\%,5\%,15\%]$, $\mathcal P[\theta] =1/3$ for any $\theta = 1,2,3$. 
We used 512 time steps and 1600 space steps in the discretization scheme. The figures show the ``continuous-stopping" approximation (with stopping rule applied at every step) as well as the results with 128, 32, 8 and 2 stopping times (e.g., in the case of 2 stopping times, the stopping rule is applied only at $0$ and at $T/2$). We observe good convergence to the continuous-stopping case: the difference between $N_S = 512$ and $N_S = 128$ stopping dates is barely visible at the level of the optimal stopping boundary and completely invisible at the level of the value function. 

Stronger volatility increases the scenario uncertainty, which has a negative impact on the option price. 

\begin{figure}
\centerline{\includegraphics[width=0.8\textwidth]{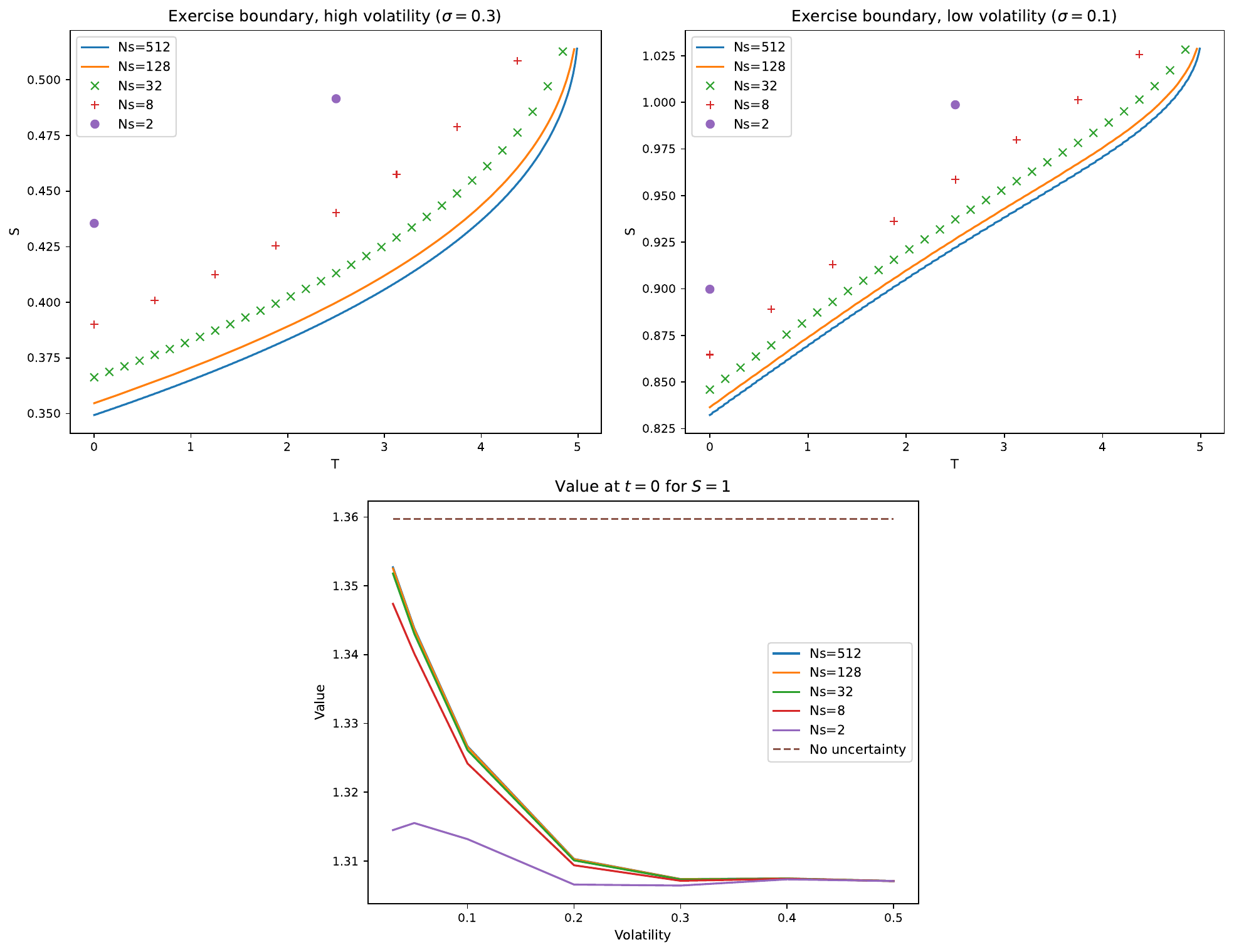}}
\caption{Convergence under no ambiguity for different numbers $N_s$ of possible stopping dates. The optimal exercise boundaries in the top panels are only shown at the dates where the stopping rule is applied.}
\label{noambig.fig}
\end{figure}

\paragraph{Ambiguity-averse investor} We now consider the ambiguity averse case, and focus an agent wishing to solve the optimization problem
$$
\sup_{\tau \in \T} v^{-1}\left( \int_{\mathcal S} d\mathcal P(\theta)v\left(\mathbb E^{\theta}\left[e^{-r\tau}S_\tau\right]\right)\right),
$$
with ambiguity function $v$ given by the power function \eqref{eq:powerlambdapositive}-\eqref{eq:powerlambdanegative}. We aim to sovlve this problem in continuous time and for a potentially infinite and uncountable set of scenarios.

By \eqref{eq:Gpower} and Theorem \ref{thm:minmax}, the problem can be rewritten as
\begin{align}
\blue{\inf_{\mathcal Q}} \,\mathbb E^{\mathcal P}\left[\left(\frac{d\mathcal Q }{d\mathcal P}\right)^{\frac{\lambda}{\lambda-1}}\right]^{{\frac{1-\lambda}{\lambda}}}\blue{\sup_{\tau \in \mathcal T}} \, \mathbb E^{\mathbb P^{\mathcal Q}}\left[e^{-r\tau }S_{\tau}\right],\label{optimsup.eq}
\end{align}
which is convenient because the second term can be determined as above for any probability measure $\mathbb P^{\mathcal Q}$. To solve this problem numerically, we discretize the inner problem using 512 time steps and 1600 space steps in the discretization scheme.}


{

In the following numerical experiment, we examine the convergence of the value function \eqref{optimsup.eq} and the optimal stopping boundary at the optimal probability measure as the number of discretization steps increases. The number of scenarios is  kept fixed is equal to three. 
Figure \ref{ambig_boundary.fig} plots the optimal exercise boundaries computed under the scenarios probabilities $(\mathcal Q[1], \mathcal Q[2], \mathcal Q[3])$ which achieve the minimum in \eqref{optimsup.eq} for different values of the ambiguity aversion parameter $\lambda$, as well as the probabilities $(\mathcal Q[1], \mathcal Q[2], \mathcal Q[3])$ at optimum as function of $\lambda$. The other parameters have the same values as in the no-ambiguity case. 

We can see that in the scenarios corresponding to higher ambiguity aversion it is optimal to exercise earlier, i.e., for higher values of the underlying, as the agents affects higher probability to adverse scenarios (low drift) and lower probability to favorable scenarios (high drift). In the presence of model ambiguity aversion, therefore, the investors are less likely to hold stocks, and require higher expected returns to hold them. 


\begin{figure}
\centerline{\includegraphics[width=1.2\textwidth]{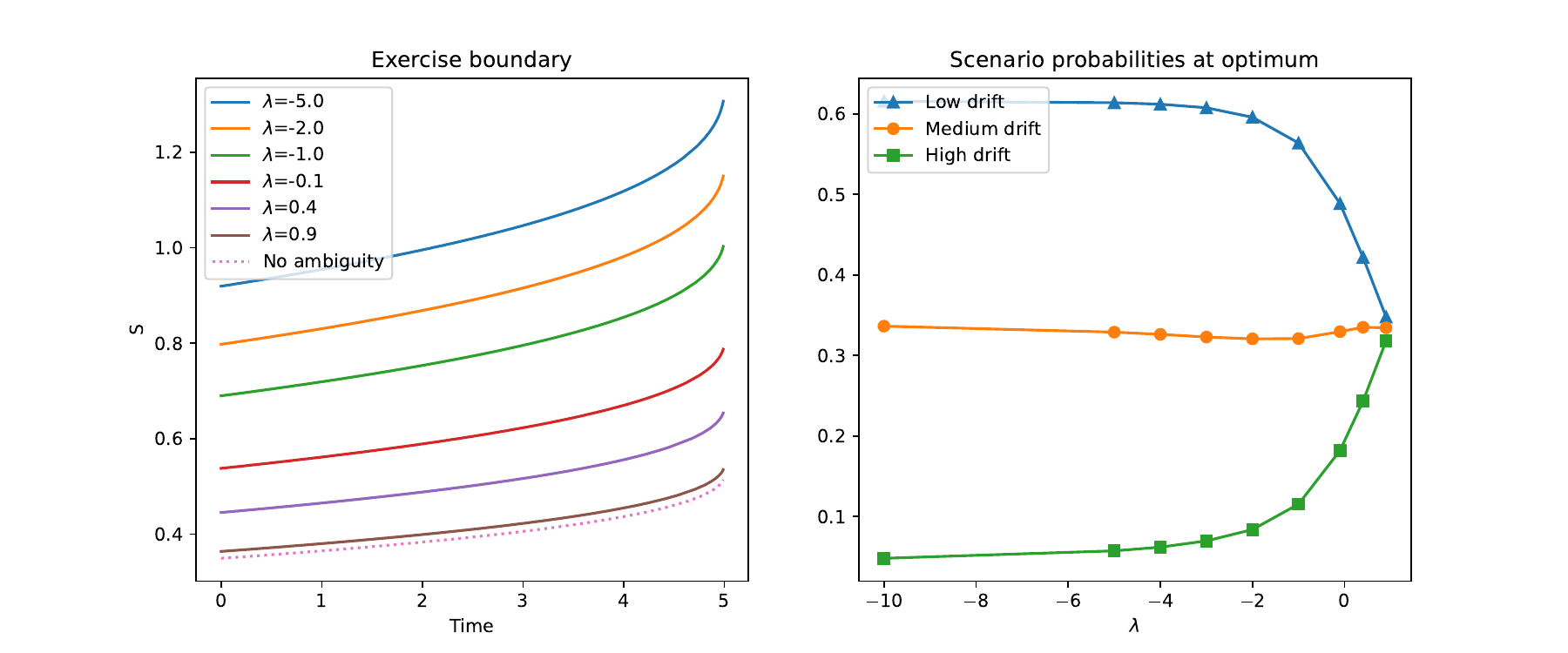}}
\caption{Exercise boundaries of \eqref{optimsup.eq}  (left) and the scenario probabilities $\lambda$ (right) as function of the model ambiguity parameter.
}
\label{ambig_boundary.fig}
\end{figure}


\begin{figure}
\centerline{\includegraphics[width=1.0\textwidth]{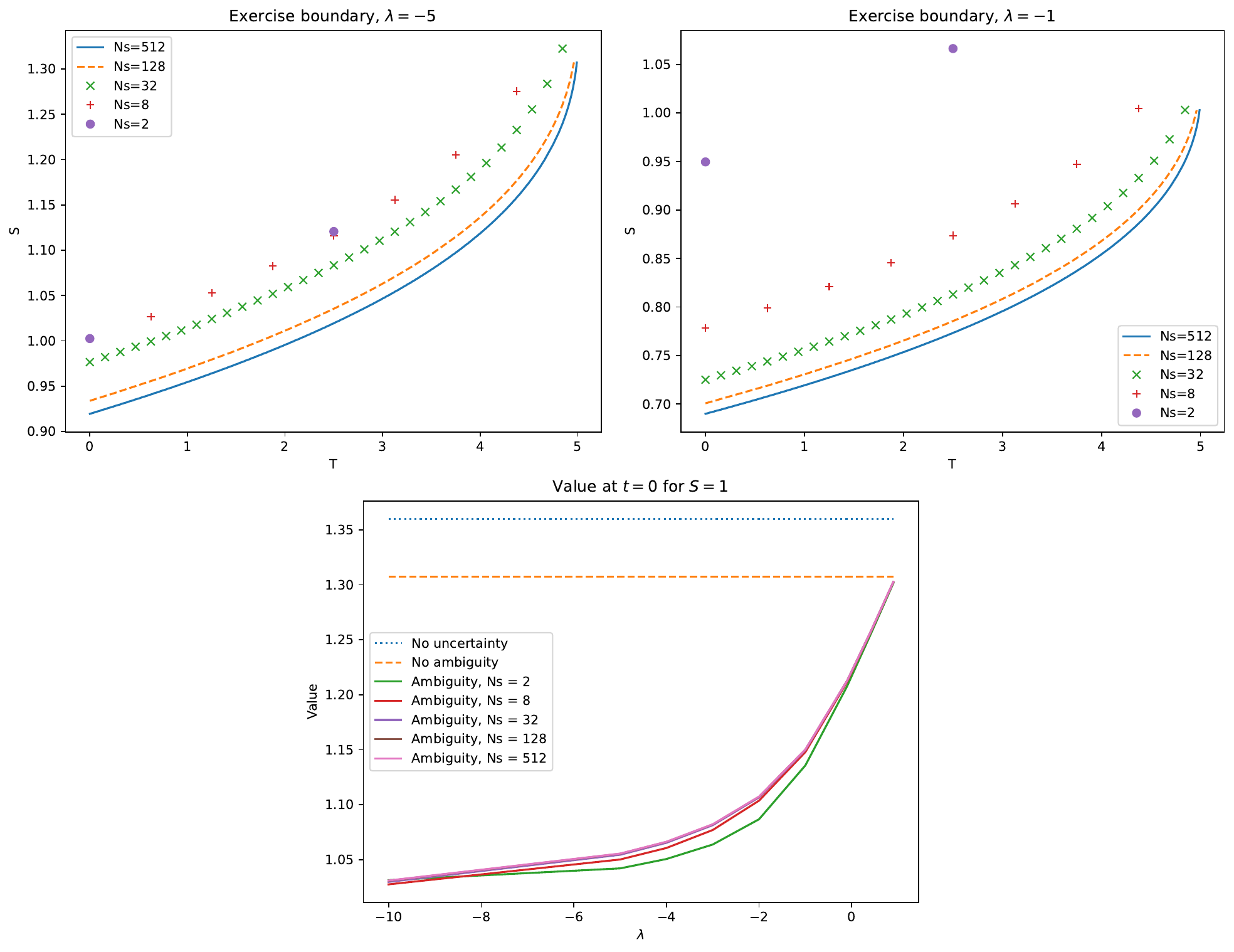}}
\caption{Convergence under ambiguity for increasing number $N_s$ of stopping dates.}
\label{ambig_convergence_times.fig}
\end{figure}

Figure \ref{ambig_convergence_times.fig} illustrates the convergence of the value and of the optimal exercise boundary in the presence of model ambiguity as the number of possible stopping times increases and for different values of the ambiguity aversion parameter. It is clear that the value decreases for strong ambiguity aversion (low $\lambda$). As in the no ambiguity case, we observe good convergence of both value and exercise boundary, although now the boundary convergence appears slower than in the no ambiguity case. In view of Theorem \ref{thm:existencesolution}, this suggests that the optimal solution to the discrete stopping time problem may be used to approximate the solution to the continuous-time problem.

In our final numerical experiment, we aim to illustrate the applicability of Theorems \ref{thm:infscenarios} and \ref{thm:minmax} beyond the case of a finite number of scenarios. We  assume that the unobservable drift $b^\theta$ can take any value in a given interval $[b_{\min}, b_{\max}]$, equipped with the uniform reference measure. This corresponds to an uncountable, one-dimensional set of scenarios. 

For numerical computation, we take $[b_{\min}, b_{\max}] = [-5\%, 15\%]$ and discretize it by $N_b$ equally spaced points, each with equal prior $1/N_b$. We then solve the corresponding discrete version of the problem for increasing values of $N_b$. This allows us to assess the convergence of the numerical solution as the discretization of the scenario space becomes finer, in analogy with the convergence tests performed earlier with respect to the discretization of the time grid.
We take same parameters as above. We again use 512 time steps and 1600 space steps in the discretization scheme and assume that each discretized time is a possible stopping date.

Figure \ref{convergence_ambiguity_scenarios.fig} shows the convergence of the value and of the optimal exercise boundary at optimal probability measure as the discretization of the drift interval becomes finer.
In particular, we take uniform grids containing $1$, $2$, $4$, $8$, $16$, $32$, $64$ and $128$ interior nodes (plus the two endpoints). We observe that the boundaries rapidly stabilize as $N_b$ increases and the grid gets refined. We also note that the exercise boundary is higher for coarser grids: we conjecture that this can be due to faster learning of the true scenario when fewer alternatives are present.

 For what regards the value, for strong ambiguity aversion the  curves are almost indistinguishable, whereas for higher $\lambda$, the learning is faster and the value is slightly higher when fewer drift scenarios are considered. 
 We observe again good convergence as the number of drift scenarios increases: the curves corresponding to $66$ and $130$ total grid points are barely distinguishable, indicating that the discretized formulation provides an accurate approximation of the continuous-scenario case.

Overall, our results indicate that the discrete formulation with finitely many stopping dates and scenarios, for which Theorem \ref{thm:existencesolution} ensures the existence of a solution, provides an accurate representation of the continuous-time, continuous-scenario problem.
}



\begin{figure}
\centerline{\includegraphics[width=1\textwidth]{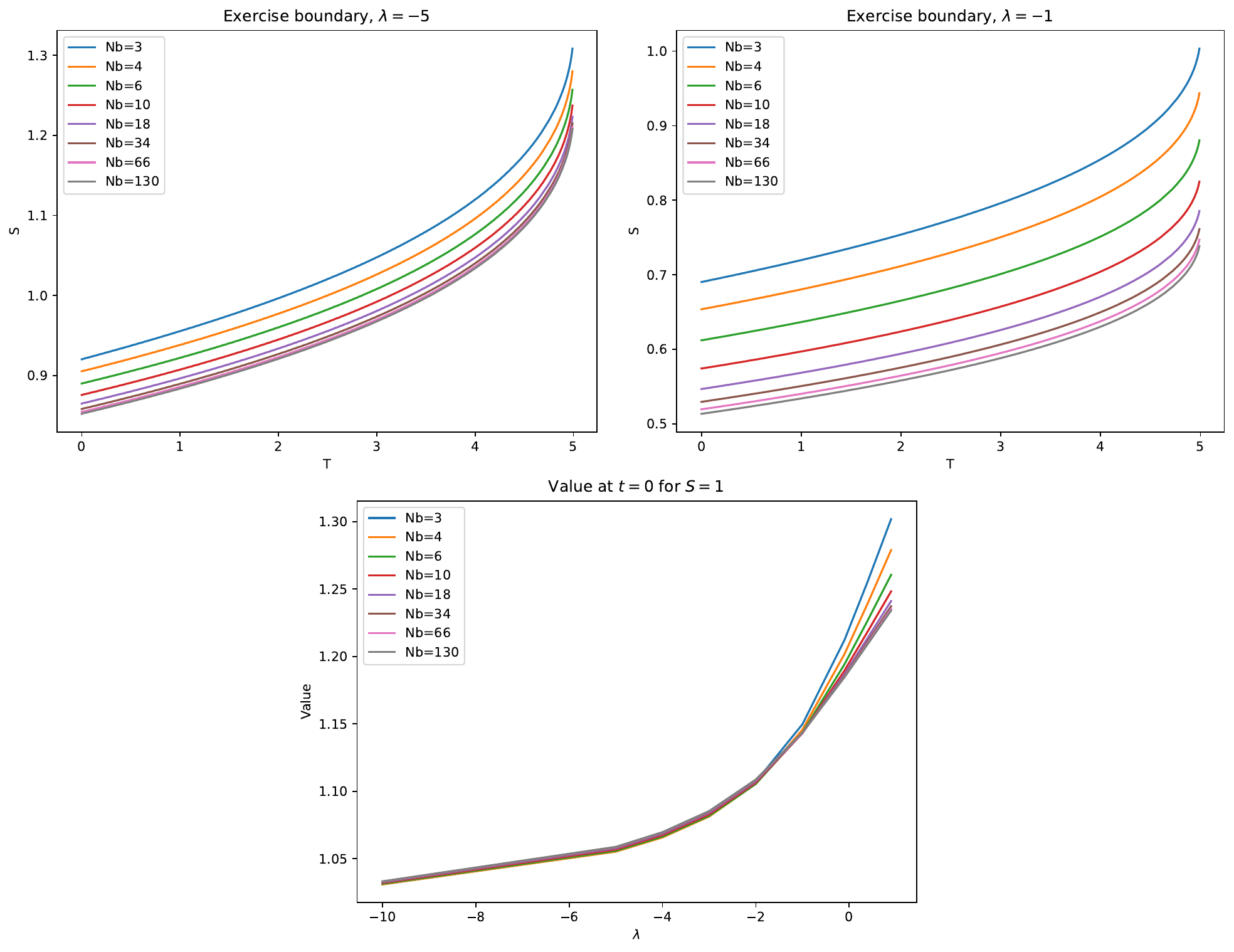}}
\caption{Convergence under ambiguity for increasing number $N_b$ of discretization points in the drift interval. }
\label{convergence_ambiguity_scenarios.fig}
\end{figure}

\section{Divestment policy under \blue{model} ambiguity \blue{on scenarios}}
\label{sec:example2}
In this section we revisit the results of \cite{flora2023green} in the presence of \blue{model ambiguity on the set of transition scenarios}. We consider the problem of an investor owning a potentially stranded asset, such as a coal-fired power plant. The future revenues of the asset  are determined by future values of risk factors (fuel prices, electricity price, carbon price) whose evolution is stochastic and whose distribution depends on the realized transition scenario. We assume there are $N$ scenarios corresponding to different climate, economic and policy assumptions, that is, $\mathcal S = \{1,\dots,N\}$. The true scenario is not known to the agent, and even the probability of occurrence of a given scenario is uncertain, that is, the agent faces both scenario uncertainty and model ambiguity. For each prior measure, the agent extracts noisy information about the scenario from the observations of a signal, and progressively updates her posterior probability of realization of each scenario based on this information. For example, if the emissions decrease at a steady rate, the agent will assume that an orderly transition scenario is more likely than a delayed transition scenario. 
In these conditions, our ambiguity-averse investor   aims to optimally choose the closure date of the project. 

Following \cite{flora2023green}, we consider this problem in a discrete-time setting.  We assume that, under scenario $i$, the \blue{$m$}-dimensional risk factor process follows
$$
X_t = \widetilde X_t + \mu^i_t,\qquad  \widetilde X_t =  \Phi \widetilde X_{t-1} + \Sigma \varepsilon_t,
$$
where  $\Sigma$ is a \blue{$m\times m$} constant volatility matrix, assumed to be non-singular, $\mu^i$ is a deterministic process, taken from integrated assessment model scenarios, $\Phi$ is a \blue{$m\times m$} matrix and $(\varepsilon_t)$ is an i.i.d.~sequence of \blue{$m$}-dimensional standard normal random vectors. 

The agent deduces scenario information from observations of a noisy signal $(S_t)_{t\geq 0}$, which is, for simplicity, assumed to be given by a scenario-dependent mean perturbed by Gaussian noise:
$$
S_t = \mu^{S,i}_t + \sigma^S \eta_t,
$$
where $\sigma^S>0$ and $\eta_t$ is an i.i.d.~sequence of standard normal random variables. 

In the absence of \blue{model ambiguity}, the investor aims to solve the following problem:
\begin{align}
\max_{\tau \in \mathcal T_T} \mathbb E^{\mathbb P^{\mathcal P}}\left[\sum_{t=1}^\tau \beta^t g(X_t) - \beta^\tau K(\tau)\right].\label{ftproblem}
\end{align}
Here $\beta$ is the discount factor, $g$ is the revenue function of the plant, $K$ is the cost of dismantling the plant, and the optimization is performed over the set $\mathcal T_T$ of all stopping times with values in $\{0,1,\dots,T\}$ in the observation filtration $\mathcal F_t = \sigma(S_s, \widetilde X_s, s=0,1,\dots,t)$. Note that we assume that the agent observes the stochastic part of the risk factor process $\widetilde X_t$ rather than the full risk factor process $X_t$ to ensure that the scenario information is extracted only from the signal. 

Denote by $\Theta$ the scenario variable and by $\boldsymbol \pi_t:=(\pi^i_t)_{i=1}^N$ the vector of posterior scenario probabilities defined by
$$
\pi^i_t = \mathbb P^{\mathcal P}[\Theta=i|\mathcal F_t],\quad i=1,\dots,N. 
$$
It can be shown (Proposition 1 in \cite{flora2023green}) that the process $(\boldsymbol\pi_t)_{t\geq 0}$ is a Markov process in the observation filtration, independent from $(\widetilde X_t)_{t\geq 0}$, and whose dynamics can be defined by
\begin{align*}
 {\pi}_t^i &= \frac{\pi^i_{t-1} e^{-\frac{\left( S_t-\mu^{i}_{S,t} \right)^2}{2 \sigma_S^2}}}{\sum_{j=1}^N \pi^j_{t-1} e^{-\frac{\left( S_t-\mu^{j}_{S,t} \right)^2}{2 \sigma_S^2}}},\quad i=1,\dots,N, 
\\
\Theta_t &= \min\{i=1,\dots,N: \sum_{j=1}^i \pi^j_{t-1} \geq U_t\},\\
S_t &= \sigma_S \eta_t + \sum_{i=1}^N \mathbf 1_{\Theta_{t} = i}\, \mu^i_{S,t},
\end{align*}
where $(U_t)_{t=0,1,\dots}$ is a sequence of independent random variables with uniform distribution on the interval $[0,1]$, independent from the sequences $(\varepsilon_t)$ and $(\eta_t)$. 
The problem \eqref{ftproblem} can therefore be solved by introducing the value function
\begin{align*}
    V(t,\widetilde{\mathbf X},  {\boldsymbol\pi}) &:= \max_{\tau\in \mathcal T_{t,T}}\mathbb E^{\mathbb P^{\mathcal P}}\left[\sum_{s=t+1}^\tau \beta^{s-t} g(\mathbf X_s) - \beta^{\tau-t} K(\tau)\, \Big|\,(\widetilde{\mathbf X}_t, {\boldsymbol \pi}_t) = (\widetilde{\mathbf X},  {\boldsymbol\pi})\right]
\end{align*}
where $\mathcal T_{t,T}$ is the set of stopping times  in the filtration of the process $(\widetilde{\mathbf X},\boldsymbol \pi)$ with values in $\{t,t+1,\dots,T\}$, and using the least squares Monte Carlo method as done in \cite{flora2023green}. 

In this paper, we are interested in the impact of \blue{model ambiguity on scenarios}. We thus consider the following optimization problem
$$
\max_{\tau \in \T_T}  \left(\sum_{\theta=1}^N \mathcal P[\theta] \left(\mathbb E^{\theta} \left[\sum_{t=1}^\tau \beta^t g(X_t) - \beta^\tau K(\tau)\right]\right)^\lambda \right)^\frac{1}{\lambda}
$$
Assuming that for every scenario $\theta\in \{1,\dots,N\}$ it holds
\begin{equation*}
\mathbb{E}^{\theta}\left[\sup_{0 \le t \le T} \left|\sum_{t=1}^\tau \beta^t g(X_t) - \beta^\tau K(\tau)\right|\right]<\infty,
\end{equation*}
we can apply Theorem \ref{thm:existencesolution} and conclude that 
\begin{multline}
\max_{\tau \in \T_T}  \left(\sum_{\theta=1}^N \mathcal P[\theta] \left(\mathbb E^{\theta} \left[\sum_{t=1}^\tau \beta^t g(X_t) - \beta^\tau K(\tau)\right]\right)^\lambda \right)^\frac{1}{\lambda} \\= \min_{\mathcal Q} \left(\mathbb E^{\mathcal P}\left[\left(\frac{d\mathcal Q}{d\mathcal P}\right)^{\frac{\lambda}{\lambda-1}} \right]\right)^{\frac{1-\lambda}{\lambda}}\max_{\tau \in \mathcal T_T}  \left(\mathbb E^{\mathbb P^{\mathcal Q}}\left[\sum_{t=1}^\tau \beta^t g(X_t) - \beta^\tau K(\tau)\right]\right)^+.\label{FTQ}
\end{multline}


For the numerical application, we consider the problem of determining the optimal closure time for a coal-fired power plant. For easy comparison, we use the same parameter values as \cite{flora2023green}. In particular, Release 3 NGFS scenarios are used, although Release 4 scenarios are already available at the time of writing. The cost of dismantling the plant, $K(\tau)$, was taken to be negative and equal to $30\%$ of the capital cost, meaning that the agent can sell the plant at any time, recovering 30\% of the initial investment. This is mainly needed becase with zero or positive dismantling cost, the maximum value of the plant is negative under the divergent transition scenario (where the carbon price is very high) and our method is not applicable. 

In the numerical simulations described below, we solve the optimization problem \eqref{FTQ}. The inner optimization (over stopping times) is performed using least squares Monte Carlo as explained in \cite{flora2023green}. One evaluation of the expectation takes about 40 seconds on a 2020 M1 MacBookPro. The outer optimization is done using the scipy implementation of COBYLA algorithm and converges, depending on the value of $\lambda$, in around 100 iterations. 

Figure \ref{uncert_ft.fig} (left graph) plots the project value as function of the uncertainty parameter $\lambda$. Project value under scenario uncertainty and model ambiguity is compared to the value without dynamic uncertainty (where all scenarios have equal probability but the agent knows the true scenario and can make the optimal stopping decision) and to the value with dynamic uncertanty but without ambiguity aversion. We see that while the impact of uncertainty is the main one, ambiguity aversion can reduce the project value by a further 10--20\% depending on the value of $\lambda$. The right graph of this figure plots the probabilities of selected scenarios as function of the undertainty parameter $\lambda$. For high ambiguity aversion (negative values of $\lambda$) the agent attributes a high probability to the worst-case scenario (here, the divergent net zero, which corresponds to the highest carbon price), and small probabilities with simular values to the other scenarios. As $\lambda$ increases, the probability of divergent net zero decreases, and the probability of other scenarios increases. 

\begin{figure}
\centerline{\includegraphics[width=\textwidth]{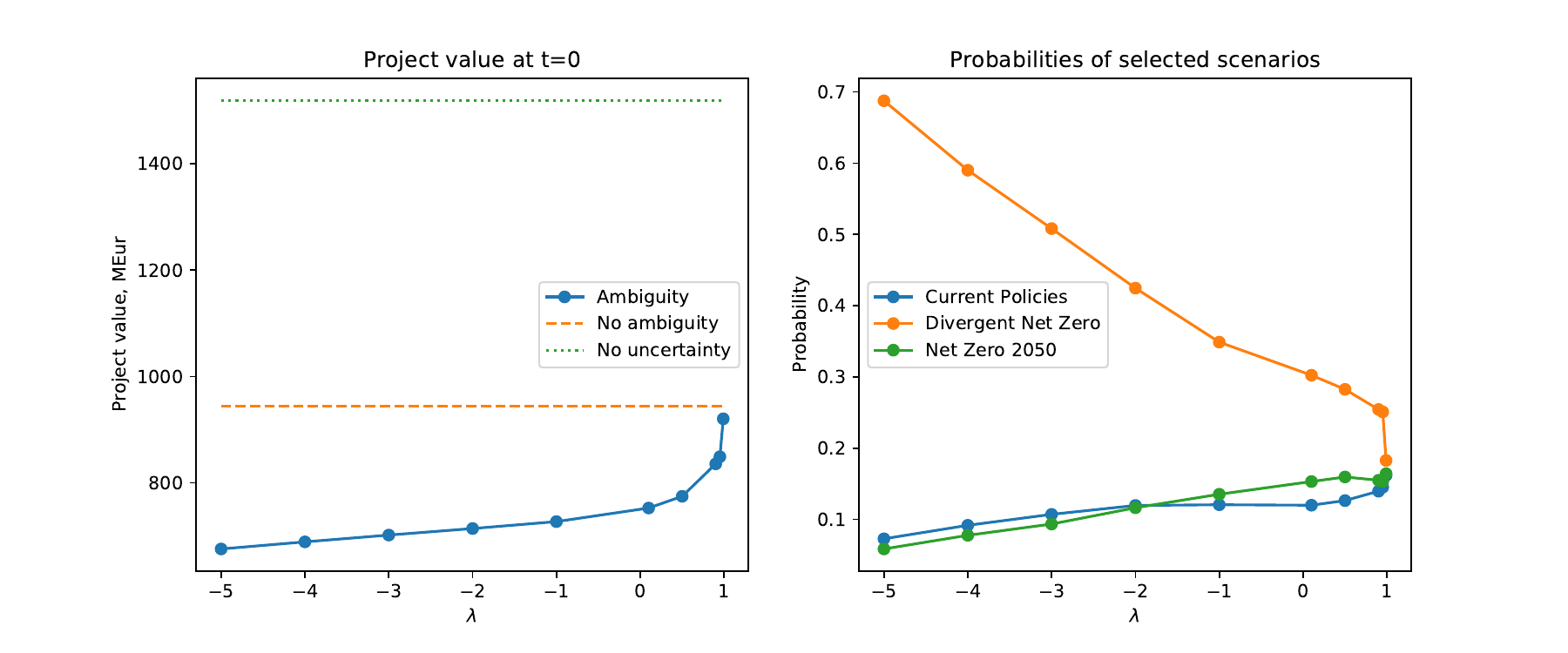}}
\caption{Project value and scenario probabilities as function of the model ambiguity parameter $\lambda$.}
\label{uncert_ft.fig}
\end{figure}

To further illustrate the impact of model ambiguity on the agent's behavior, Figure \ref{taufreq.fig} shows the distribution of plant closure times under different assumptions on scenario uncertainty and \blue{model} ambiguity. In each graph, the height of the bar corresponds to the proportion of trajectories where the exit decision is taken in a given year, and the color / pattern corresponds to the true scenario. For instance, in the top graph, when the scenario is immediately revealed to the agent and there is no ambiguity we see that, for example, when the true scenario is ``Current Policies", the plant always operates for its maximum lifetime ($T=30$), and when the true scenario is ``Divergent net zero" or ``Net Zero 2050", the plant is closed immediately. The small uncertainty in the case of NDCs scenario is due to other stochastic factors. The middle graph corresponds to the situation when the scenario is uncertain but there is no model ambiguity. In that case, the distribution of stopping times for each scenario is more spread out, for example, if the true scenario is ``Divergent net zero", the plant is closed between 2nd and 8th year, whereas it would have been optimal to stop immediately. This is because one needs to wait to acquire information about the scenario before taking the stopping decision. Finally, the bottom graph corresponds to the situation when the prior probabilities of scenarios are ambiguous and the agent is ambiguity averse. In this case, the agent favors early stopping, because a high subjective probability is attributed to the worst case scenario. 

{
For policymakers, these results show that ambiguity aversion further depresses the valuation of potentially stranded assets relative to a no-ambiguity benchmark, which is itself lower than in the full-information case in which the transition path is known. Under ambiguity, stranding occurs earlier and with greater losses, implying that models which ignore ambiguity understate both timing and severity. Ambiguity aversion may therefore accelerate exit from carbon-intensive assets, but it can also amplify transition risks and weaken financial stability.

These results speak directly to the design and communication of transition-scenario databases issued by international bodies (e.g., NGFS). While such databases enhance risk assessment, they also shape beliefs about scenario uncertainty: depending on whether they narrow or widen perceived ambiguity, they can mitigate or exacerbate risk. Careful specification and transparent communication of uncertainty—including ranges, model dispersion, and confidence—thus become policy-relevant levers.

} 

\begin{figure}
\centerline{\includegraphics[width=\textwidth]{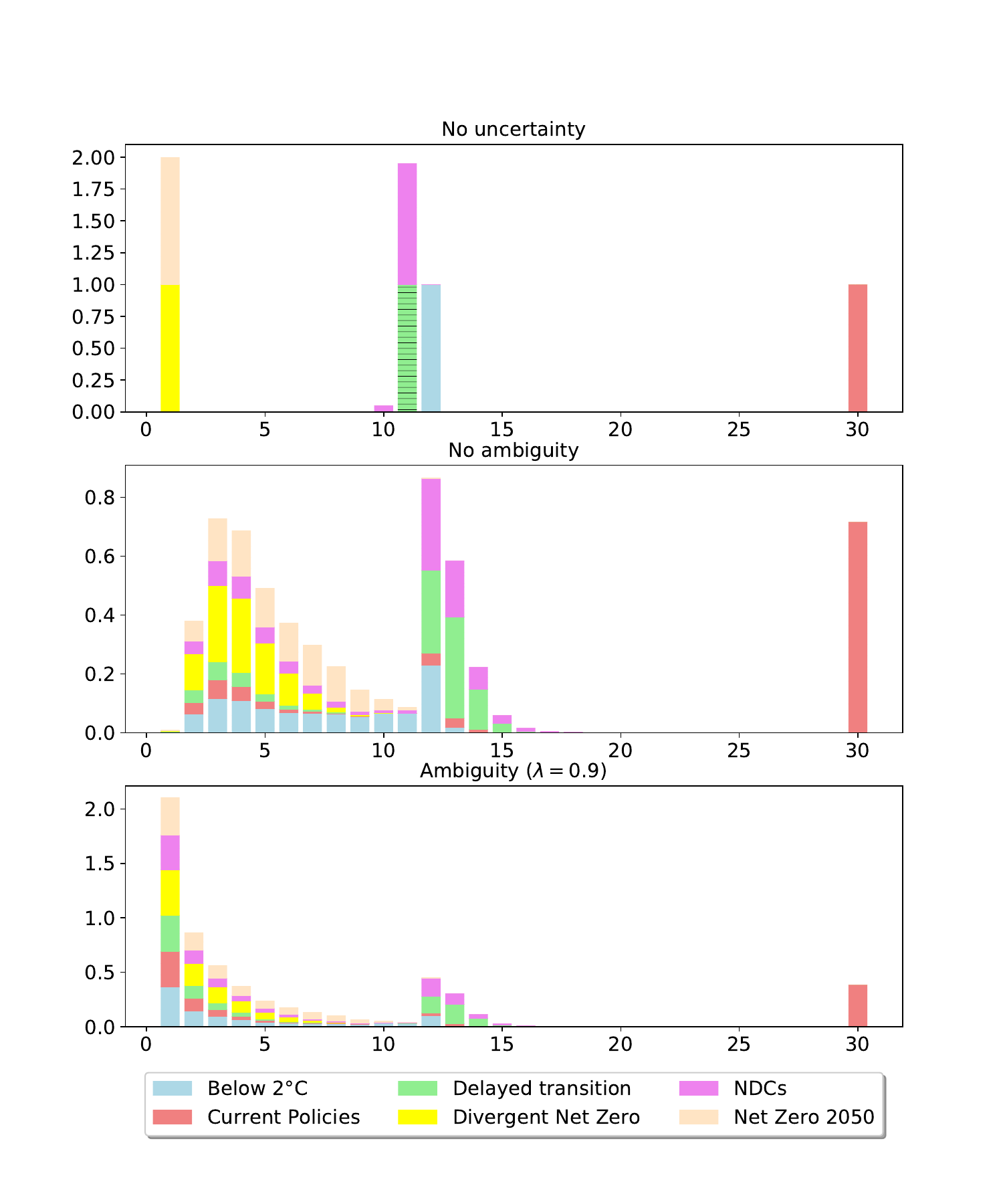}}
\caption{Distribution of plant closure times in the absence of scenario uncertainty (top graph), in the presence of uncertainty but in absence of ambiguity (middle graph) and in the presence of ambiguity (bottom graph). The color/pattern corresponds to the true scenario. }
\label{taufreq.fig}
\end{figure}



\appendix
\section{Appendix}
Throughout this appendix, we assume that the function $v$, which defines the ambiguity attitude of the investor, satisfies the assumptions of Proposition \ref{prop:drapeau}, so that there exists a function $R$ with the properties listed in this proposition. 

\subsection{Proof of Proposition \ref{boundary.prop}}
\label{boundary.proof}
{\begin{lemma}\label{gamma.lm}
Let the prior distribution $m_\mu$ have bounded support and let $\mu_1$ and $\mu_2$ be, respectively, the left and right endpoints of the support. Then the coefficient $\Gamma_m$ is increasing, differentiable in $t$ and $x$ with bounded derivatives, and satisfies
\begin{align*}
|\Gamma_m(t,x)|&\leq \max\{|\mu_1|,|\mu_2|\},\quad (t,x)\in [0,T]\times \mathbb R. \\
\lim_{x\to +\infty}\Gamma_m(t,x)& = \mu_2\\
\lim_{x\to -\infty}\Gamma_m(t,x)& = \mu_1,\\
|\partial_t \Gamma_m(t,x)|\leq C \partial_x \Gamma_m(t,x) 
\end{align*}
for some constant $C<\infty$,
where the convergence is uniform in $t$. 
\end{lemma}}
\begin{proof}{
From \eqref{eq:generalGm}, the derivative of $\Gamma_m$ is 
\begin{align*}
\partial_x\Gamma_m (t,x)&= \frac{1}{\sigma^2}\frac{ \int z^2 m_\mu(dz) \mathcal E(t,x,z) \int m_\mu(dz) \mathcal E(t,x,z) - \left(\int z m_\mu(dz) \mathcal E(t,x,z)\right)^2}{ \left(\int m_\mu(dz) \mathcal E(t,x,z)\right)^2}\\
& = \frac{1}{\sigma^2} \left\{\int z^2 \bar m_\mu^{t,x}(dz) -\Big(\int z \bar m_\mu^{t,x}(dz)\Big)^2\right\}
\end{align*}
with
$$
\mathcal E(t,x,z) = \exp\left(-\frac{t z^2}{2\sigma^2} + \frac{z(x-X_0 )}{\sigma^2}\right)\quad \text{and}\quad 
\bar m_\mu^{t,x}(A):= \frac{\int_A m_\mu(dz) \mathcal E(t,x,z) }{\int_{\mathbb R} m_\mu(dz) \mathcal E(t,x,z)},
$$
$\quad A\in \mathcal B(\mathbb R)$.
It is then clear that 
$$|\Gamma_m(t,x)|\leq \frac{\max\{|\mu_1|,|\mu_2|\}}{\sigma^2}\quad \text{and}\quad |\partial_x\Gamma_m(t,x)|\leq 2\frac{\max\{|\mu_1|,|\mu_2|\}^2}{\sigma^2}.
$$ 
Moreover, by Cauchy-Schwarz inequality, $\partial_x\Gamma(t,x)\geq 0$ which means that $\Gamma$ is increasing in $x$. 

Further, the derivative of $\Gamma_m$ with respect to $t$ is given by
$$
\partial_t \Gamma_m(t,x) =\frac{1}{2\sigma^2}\left\{-\int z^3 \bar m_\mu^{t,x}(dz)+\int z^2 \bar m_\mu^{t,x}(dz)\int z \bar m_\mu^{t,x}(dz)\right\},
$$
from which the remaining properties follow easily. }
\end{proof}

\begin{proof}[Proof of Proposition \ref{boundary.prop}] {In this proof, we will simply write $\mathbb E$ for the expectation $\mathbb E^{\mathbb P^{\mathcal P}}$. 
\paragraph{Existence}Let $h>0$. Then, for $s\geq t$, 
$$
X^{t,x+h}_s - X^{t,x}_s = h + \int_t^s (\Gamma(u,X^{t,x+h}_u) - \Gamma(u,X^{t,x}_u)) du,
$$
so that 
\begin{align}
h\leq X^{t,x+h}_s - X^{t,x}_s \leq h e^{K(s-t)},\label{boundx}
\end{align}
with $K = 2\frac{\max\{|\mu_1|,|\mu_2|\}^2}{\sigma^2}$.

Recall that
$$
v(t,x) = \sup_{t\leq \tau \leq T} \mathbb E [e^{X^{t,x}_\tau }]
$$
Then,
\begin{align*}
v(t,x+h) - e^{x+h}  - (v(t,x)-e^x) & =  \sup_{t\leq \tau \leq T} \mathbb E [e^{X^{t,x+h}_\tau }] - e^{x+h} - \sup_{t\leq \tau \leq T} \mathbb E [e^{X^{t,x}_\tau }] + e^x\\
&\geq (e^h-1)(v(t,x) - e^x)\geq 0,
\end{align*}
which shows that $v(t,x)-e^x$ is increasing in $x$. Moreover,
\begin{align*}
|v(t,x+h)  - v(t,x)|&\leq (e^{he^{K(T-t)}}-1) v(t,x)\\
v(t,x) &\leq \sup_{t\leq \tau \leq T} \mathbb E [e^{x+\sigma \widehat W_{\tau-t}+ (T-t){\max\{|\mu_1|,|\mu_2|\}} }],
\end{align*}
which shows that $v(t,x)$ is 
continuous in $x$. 
This implies that for each $t\in [0,T)$, the set $\{x:v(t,x)=e^x\}$ is closed and we can define $b(t) = \max\{x:v(t,x)=e^x\}\in [-\infty,\infty]$, which satisfies $\{v(t,x)=e^x\} = \{x\leq b(t)\}$. 

\paragraph{Upper bound}For $0\leq t < s<T$, 
$$
v(t,x)\geq \mathbb E[e^{X^{t,x}_s}] = e^x + e^x\mathbb E\left[\int_t^s \{\Gamma(u,X^{t,x}_u)+\sigma^2/2\}e^{X^{t,x}_u-x}du\right].
$$
By dominated convergence,
$$
\lim_{s\downarrow t}\frac{1}{s-t}\mathbb E\left[\int_t^s \{\Gamma(u,X^{t,x}_u)+\sigma^2/2\}e^{X^{t,x}_u-x}du\right] = \Gamma(t,x)+\frac{\sigma^2}{2}.
$$
Therefore, for each $t\in [0,T)$ the set $\{x:\Gamma(t,x)+\sigma^2/2 >0\}$ belongs to the continuation region. By Lemma \ref{gamma.lm} and the assumption of the proposition, the exercise region is bounded from above by a constant, so that $b(t)<C$ for $t\in[0,T)$ for some $C<\infty$. 


\paragraph{Upper semicontinuity}
To prove upper semicontinuity of the boundary, let us first show that $v(t,x)$ is continuous in $t$. Fix $t>0$ and $h>0$ with $t+h<T$. On the one hand, we have:
$$
v(t,x)\geq \sup_{t+h\leq \tau \leq T}\mathbb E[e^{X^{t,x}_\tau}] = \mathbb E[v(t+h,X^{t,x}_{t+h})] \geq v(t+h,x) -K\mathbb E[|X^{t,x}_{t+h}-x|],
$$
where $K$ is the Lipschitz constant of $v(t+h,\cdot)$. On the other hand, by standard arguments, $\mathbb E[|X^{t,x}_{t+h}-x|] = O(\sqrt{h})$. Further, 
\begin{align*}
v(t,x) &= e^{x} + \sup_{t\leq \tau \leq T} \mathbb E\left[\int_t^\tau e^{X^{t,x}_s}(\Gamma(s,X^{t,x}_s) + \sigma^2/2)ds\right]\\ &\leq e^x + \mathbb E\left[\int_t^{t+h} e^{X^{t,x}_s}\mu_2 ds\right] + \sup_{t+h\leq \tau \leq T} \mathbb E\left[\int_{t+h}^\tau e^{X^{t,x}_s}(\Gamma(s,X^{t,x}_s) + \sigma^2/2)ds\right]\\
& = \mathbb E[e^{x} - e^{X^{t,x}_{t+h}}] + \mathbb E\left[\int_t^{t+h} e^{X^{t,x}_s}\mu_2 ds\right] + \mathbb E[v(t+h,X^{t,x}_{t+h})].
\end{align*}
Using standard arguments for the first two terms, and the Lipschitz continuity of $v$ for the last one, we see that $v(t,x)\leq v(t+h,x) + O(\sqrt{h})$. Thus $v$ is $\frac{1}{2}$-Hölder continuous in $t$.

Assume that, for some $t\in [0,T)$ and $x\in \mathbb R$, $v(t,x)>e^x$. Then, by continuity of $v$, there is a neighborhood of $(t,x)$ such that for all $(t',x')$ belonging to this neighborhood, $v(t',x')>e^{x'}$. This implies that the continuation region is open, and therefore, the exercise boundary $b$ is upper semicontinuous. 
\paragraph{Lower semicontinuity}
Let $t<T$ and assume that $v(t,x) = e^x$. This means that 
$$
e^x\geq \sup_{t\leq \tau \leq T-h} \mathbb E[e^{X^{t,x}_\tau}]
$$
for any $h$ with $0\leq h\leq T-t$. 
Define the process $\widetilde X^{t,x}$ as the solution of the stochastic differential equation
$$
d\widetilde X^{t,x}_s = \Gamma(s+h,\widetilde X^{t,x}_s - C h) ds + \sigma dW_s,
$$
starting from $x$ at time $t$, where $C$ is the constant defined in Lemma \ref{gamma.lm}. By this lemma and the comparison theorem for SDEs, $\widetilde X^{t,x}_s \leq X^{t,x}_s$ for $s\geq t$. Therefore,
$$
e^x \geq \sup_{t\leq \tau \leq T-h} \mathbb E[e^{\widetilde X^{t,x}_\tau}] = \sup_{t+h\leq \tau \leq T} \mathbb E[e^{Ch+X^{t+h,x-Ch}_\tau}]
$$
so that $v(t+h,x-Ch) = e^{x-Ch}$ and the point $(t+h,x-Ch)$ belongs to the exercise region. 

Finally, let $t<T$, assume that $v(t,x')=e^{x'}$ and let $x<x'$. By the first part of this proof and the continuity of $\Gamma_m$, one can find $h>0$ and $\mu$ with $\mu + \frac{\sigma^2}{2}<0$ such that $\Gamma(s,z)<\mu$ for all $(s,z)\in [t-h,t]\times (-\infty, x]$. By Hölder continuity of $v$ in the first variable, for $(s,z)$ in this set, $v(s,z)\leq e^{z} + C\sqrt{h}$ for some constant $C>0$. 

Define the process $\overline X^{s,z}$ by $\overline X^{s,z}_r = \mu (r-s)+\sigma (W_r-W_s)$ and the stopping time $\tau^{s,z}_x:=\inf\{r>s: \overline X^{s,z}_r = x\}$. Then,
$$
v(s,z) \leq \bar v(s,z):=\sup_{\tau\geq s}\mathbb E[(e^x+C\sqrt{h})\mathbf 1_{\tau\geq \tau^{s,z}_x} + e^{{\bar{X}}^{s,z}_\tau}\mathbf 1_{\tau <\tau^{s,z}_x}].
$$
Note that here the stopping time is unbounded and therefore $\bar v(s,z)$ does not depend on $s$. This function is the solution of the variational inequality
$$
\max\left\{\frac{\sigma^2}{2} \frac{\partial^2 \bar v}{\partial z^2} + \mu \frac{\partial \bar v}{\partial z},e^z - \bar v\right\} = 0
$$
on $(-\infty,x)$, with boundary condition $\bar v(x) = e^x + C\sqrt{h}$. To solve it, we look for $x^*<x$ with 
\begin{align*}
&\frac{\sigma^2}{2} \frac{\partial^2 \bar v}{\partial z^2}(z) + \mu \frac{\partial \bar v}{\partial x}(z) = 0,\quad z>x^*, \\
&\bar v(x) = e^x + C\sqrt{h},\quad \bar v(x^*) = e^{x^*},\quad \bar v'(x^*) = e^{x^*}. 
\end{align*}
The solution of the differential equation writes
$$
\bar v(z) = C_1 + C_2 e^{-\frac{2\mu}{\sigma^2}z}.
$$
Substituting this into the boundary condition, we get:
$$
C_1 + C_2 e^{-\frac{2\mu}{\sigma^2}x} = e^x + C\sqrt{h},\quad C_1 + C_2 e^{-\frac{2\mu}{\sigma^2}x^*} = e^{x^*},\quad -\frac{2\mu C_2}{\sigma^2} e^{-\frac{2\mu}{\sigma^2}x^*}= e^{x^*},
$$
and solving for $x^*$, we find
$$
1-\frac{\sigma^2}{2|\mu|} + \frac{\sigma^2}{2|\mu|} e^{\frac{2|\mu|}{\sigma^2}(x-x^*)} = e^{x-x^*}(1+ C\sqrt{h}e^{-x})
$$
It is easy to see that this equation admits a solution {$x^*$} which satisfies {$x^*<x$} and converges to {$x$} as $h\to 0$. Since $x$ was arbitrarily close to $x'$, we conclude that $(s,x)$ belongs to the exercise region for any $x<x'$, provided that $s<t$ is sufficiently close to $t$. Together, the two parts of this step show that the exercise boundary is lower semicontinuous, and since it has already been shown to be upper semicontinuous, the proof of continuity is complete. }

\end{proof}

\subsection{Proof of Theorems \ref{thm:existencesolution} and \ref{thm:infscenarios}}\label{sec:proofexistence}

\blue{We proceed in two main steps. First, we prove Theorem \ref{thm:existencesolution} under the assumption of a finite scenario set. 
The key ingredient is the representation of randomized stopping times introduced in \eqref{eq:PTinfty} and the associated topological structure, both taken from \cite{belomestny2016optimal}. The chosen topology ensures the compactness of the space, which allows us to apply the classical minimax theorem of \cite{sion1958general} to establish Lemma \ref{saddle.lm} for randomized stopping times. Moreover, the finiteness of the set of scenarios enables the construction of a mapping from randomized to ordinary stopping times which preserves the expected payoff under every scenario, as shown in Lemma \ref{thin.lm}, which lead to the proof of Theorem \ref{thm:existencesolution}.
}

\blue{In the second step, we prove Theorem \ref{thm:infscenarios}, which extends Theorem \ref{thm:existencesolution} to the case of a compact, possibly uncountable scenario space. This is achieved by using Assumption \ref{cont.ass}, which provides the continuity properties required to pass to the limit in the number of scenarios.}


We start with some preliminary notation and two lemmas.

Let $\mathcal{P}_{\mathbb{T}}$ be the set of all $\left(A_t, t\in \mathbb{T}\right)$ satisfying $A_t \in \mathcal{F}_t$ for $t \in\mathbb{T}$  as well as $\mathbb{P}\left(A_t \cap A_s\right)=0$ for $t \ne s$, and $\mathbb{P}\left(\cup_{t\in \mathbb{T}} A_t\right)=1$. 
{There is a one-to-one correspondence between the set $\mathcal{P}_{\mathbb{T}}$ and the set of stopping times $\tau$ with values in $\mathbb{T}$ by taking $A_t = \{\tau=t\}$.}

Further, define
\begin{align}
\mathcal{P}_{\mathbb{T}}^{\infty}:=\Bigg\{&(f_t,t\in \mathbb{T}), \ f_t \in L^{\infty}\left(\Omega, \F_t, \mathbb{P} |_{\F_t}\right), \ f_t \ge 0 \text{ } \mathbb{P}-\text{a.s. } \forall t\in \mathbb{T}, \notag\\ &\sum_{t\in \mathbb{T}} f_t = 1 \text{ } \mathbb{P}-\text{a.s. } \Bigg\}.\label{eq:PTinfty}
\end{align}
There is a mapping of the set of randomized stopping times on $\mathbb{T}$ onto the set $\mathcal{P}_{\mathbb{T}}^{\infty}$. Recall that  a randomized stopping time with respect to a filtered probability space $\left(\Omega, \mathcal{F},\mathbb F:=\left(\mathcal{F}_t\right)_{t\in \mathbb T}, \mathbb{P}\right)$ is a mapping $\gamma: \Omega \times[0,1] \rightarrow\mathbb T$, nondecreasing and left-continuous in the second component and such that $\gamma(\cdot, u)$ is a $\mathbb F$-stopping time  for any $u \in[0,1]$. 
If $\gamma(\cdot,\cdot)$ is a randomized stopping time with values in $\mathbb T$, we take
$$
f_t(\omega)= \int_0^1 \mathbf 1_{\gamma(\omega,v) = t}\, dv,\quad \omega \in \mathcal F_t,\  t\in \mathbb T. 
$$

Let $\prod_{t\in \mathbb T} \sigma\left(L_t^{\infty}, L_t^1\right)$ be the product topology of $\sigma\left(L_t^{\infty}, L_t^1\right)$, $t\in \mathbb T$, on $\prod_{t\in \mathbb T} L^{\infty}\left(\Omega, \mathcal{F}_t,\left.\mathbb{P}\right|_{\mathcal{F}_t}\right)$, where $\sigma\left(L_t^{\infty}, L_t^1\right)$ denotes the weak* topology on \linebreak$L^{\infty}\left(\Omega, \mathcal{F}_t,\left.\mathbb{P}\right|_{\mathcal{F}_t}\right)$.
Then, from \blue{Banach-Alaoglu} theorem, $\mathcal{P}_{\mathbb T}^{\infty}$ is compact w.r.t. $\prod_{t\in \mathbb T} \sigma\left(L_t^{\infty}, L_t^1\right)$.  

\begin{lemma}\label{thin.lm}
For any $\left(f_t,t\in \mathbb T\right) \in \mathcal{P}_{\mathbb T}^{\infty}$, and any subset of $\mathcal S_N \subseteq \mathcal S$ containing $N$ elements, there exist $\left(A^N_t,t\in \mathbb T\right) \in \mathcal{P}_{\mathbb T}$ such that 
$$
\mathbb E^\theta [f_t Y_t] = \mathbb E^\theta[\mathbf{1}_{A^N_t}Y_t],\quad t\in \mathbb T, \quad \theta \in \mathcal S_N. 
$$
\end{lemma}
\begin{proof}
Recall that for a fixed probability space $(\Omega,\mathcal F, \mathbb P)$ a subset $M\subseteq L^1(\Omega,\mathcal F, \mathbb P)$ is called thin if for any $A\in \mathcal F$ with $\mathbb P(A)>0$, there is some nonzero $g\in L^\infty(\Omega,\mathcal F, \mathbb P)$ vanishing outside $A$ and satisfying $\mathbb E[g\cdot Z] = 0$ for any $Z\in M$. 
The set
$$
\left \{\mathbb{E}^\theta\left[\mathbf{1}_{A}Y_{s}\bigg|\F_t\right], \theta\in \mathcal S_N \right \}
$$
is a thin subset  of  $ L^{1}\left(\Omega, \F_t, \mathbb{P} |_{\F_t}\right)$, for any $A \in  \F_T$ and any $t,s \in \mathbb T$ with $t \le s$ as a finite set on a non-atomic probability space, see \cite{kingman1968theorem}, p. 348. Together with Proposition C.3 of  \cite{belomestny2016optimal}, this implies the statement of the lemma. 
\end{proof}

Consider the mapping
$$
L: \mathcal M (\mathcal S) \times \mathcal P^\infty_{\mathbb T},\quad (\mathcal Q, (f_t,t\in \mathbb T))\mapsto R\left(\mathcal Q, \sum_{t\in \mathbb T} \int_{\mathcal S} d\mathcal Q(\theta) \mathbb E^\theta[f_t Y_t]\right)
$$
\begin{lemma}\label{saddle.lm}
Assumptions \ref{ass:equivalence}, \ref{ass:integrability} and \ref{continuity} be satisfied. Then there exists a saddle point $(\mathcal Q^*, (f^*_t,t\in \mathbb T))$ such that 
$$
L(\mathcal Q^*,(f_t,t\in \mathbb T)) \leq L(\mathcal Q^*,(f^*_t,t\in \mathbb T)) \leq L(\mathcal Q,(f^*_t,t\in \mathbb T))
$$
\end{lemma}
\blue{for any $\mathcal{Q} \in \mathcal{M}(\mathcal{S})$ and any $(f_t,t\in \mathbb T) \in P^\infty_{\mathbb T}$.}
\begin{proof}
We want to apply Theorem 3.4 of \cite{sion1958general} to $L$. The sets $\mathcal M (\mathcal S)$ and $\mathcal P^\infty_{\mathbb T}$ are clearly convex. 
Since $R$ is nondecreasing in the second argument, it is also quasi-concave with respect to it, so that $L(\mathcal Q,\cdot)$ is quasi-concave for any $\mathcal Q \in \mathcal M(\mathcal S)$. Point (ii) of Proposition \ref{prop:drapeau} implies quasi-convexity of $L(\cdot,(f_t,t\in \mathbb T))$ for any $(f_t,t\in \mathbb T)\in \mathcal P^\infty_{\mathbb T}$, and thus $L$ is quasi-concave-convex.

Moreover, the function  $L(\mathcal Q,\cdot)$ is upper semicontinuous by point (i) of Proposition \ref{prop:drapeau} and since 
$$
(f_t,t\in \mathbb T) \mapsto \sum_{t\in \mathbb T} \int_{\mathcal S} d\mathcal Q(\theta) \mathbb E^\theta[f_t Y_t]
$$
 is continuous for all $\mathcal Q \in \mathcal M(\mathcal S)$ by Assumption \ref{ass:integrability}. 
 
 It remains to show that $L(\cdot,(f_t,t\in \mathbb T))$ is lower semicontinuous for all $(f_t,t\in \mathbb T)\in \mathcal P^\infty_{\mathbb T}$. Assume first that we are in the context of Assumption \ref{continuity} (i). Then, for any sequence $(\mathcal Q^n)$ converging to a limit $\mathcal Q \in \mathcal M(\mathcal S)$ and for any index $k\in \mathbb N$, 
\begin{align*}
\liminf_n L(\mathcal Q^n,(f_t,t\in \mathbb T)) &= \lim_n \inf_{m\geq n} R\left(\mathcal Q^m,\sum_{t\in \mathbb T} \int_{\mathcal S} d\mathcal Q^m(\theta) \mathbb E^\theta[f_t Y_t]\right) \\ &\geq  \lim_n \inf_{m\geq n} R\left(\mathcal Q^m,\inf_{j\geq k}\sum_{t\in \mathbb T} \int_{\mathcal S} d\mathcal Q^j(\theta) \mathbb E^\theta[f_t Y_t]\right) \\
&\geq  R\left(\mathcal Q,\inf_{j\geq k}\sum_{t\in \mathbb T} \int_{\mathcal S} d\mathcal Q^j(\theta) \mathbb E^\theta[f_t Y_t]\right),
\end{align*}
where the first inequality holds because $R^+$ is nondecreasing in $s$.
Now, making $k$ tend to $+\infty$, by continuity we obtain
$$
\liminf_n L(\mathcal Q^n,(f_t,t\in \mathbb T)) \geq L(\mathcal Q,(f_t,t\in \mathbb T)).
$$
Under Assumption \ref{continuity} (ii) a similar argument can be used since $$\sum_{t\in \mathbb T}\int_{\mathcal S} d\mathcal Q(\theta) E^{\theta}[f_t Y_t]>0.$$

Since both the set of probabilities $\mathcal M (\mathcal S)$ and the set $\mathcal P^\infty_{\mathbb T}$ are compact and the function $L$ is u.s.c.-l.s.c., the infimum and the supremum in Sion's theorem  are attained and there exists a saddle point.
\end{proof}

\begin{proof}[Proof of Theorem \ref{thm:existencesolution}]
Since the set of scenarios is finite, by Lemma \ref{thin.lm}, we can find $(A^*_t,t\in \mathbb T)\in \mathcal P_{\mathbb T}$ such that 
$$
\mathbb E^\theta[f^*_t Y_t] = \blue{\mathbb E^\theta[\mathbf{1}_{A^*_t}Y_t]},\quad t\in \mathbb T,\ \theta = 1,\dots,N,
$$
so that 
$$
L(\mathcal Q,(f^*_t,t\in \mathbb T)) = L(\mathcal Q,(\mathbf 1_{A^*_t},t\in \mathbb T)),\quad \forall \mathcal Q \in \mathcal M(\mathcal S),
$$
which means that for all $(A_t,t\in \mathbb T)\in \mathcal P_{\mathbb T}$
$$
L(\mathcal Q^*,(\mathbf 1_{A_t},t\in\mathbb T)) \leq L(\mathcal Q^*,(\mathbf 1_{A^*_t},t\in \mathbb T)) \leq L(\mathcal Q,(\mathbf 1_{A^*_t},t\in \mathbb T)).
$$
\end{proof}
\begin{proof}[Proof of Theorem \ref{thm:infscenarios}]
We now would like to prove that
$$
\sup_{\mathcal P_{\mathbb T}} \inf_{\mathcal Q \in \mathcal M(\mathcal S)} L(\mathcal Q,(\mathbf 1_{A_t},t\in \mathbb T) ) =  \inf_{\mathcal Q \in \mathcal M(\mathcal S)} \sup_{\mathcal P_{\mathbb T}} L(\mathcal Q,(\mathbf 1_{A_t},t\in \mathbb T) )
$$
without assuming that the set of scenarios is finite. 
Clearly,
$$
\sup_{\mathcal P_{\mathbb T}} \inf_{\mathcal Q \in \mathcal M(\mathcal S)} L(\mathcal Q,(\mathbf 1_{A_t},t\in \mathbb T) ) \leq  \inf_{\mathcal Q \in \mathcal M(\mathcal S)} \sup_{\mathcal P_{\mathbb T}} L(\mathcal Q,(\mathbf1_{A_t},t\in \mathbb T) ),
$$
and it remains to prove the opposite inequality. By Lemma \ref{saddle.lm},
\begin{align*}
\sup_{ \mathcal P^\infty_{\mathbb T}} \inf_{\mathcal Q \in \mathcal M(\mathcal S)} L(\mathcal Q,(f_t,t\in \mathbb T) ) &\geq  \inf_{\mathcal Q \in \mathcal M(\mathcal S)} \sup_{\mathcal P^\infty_{\mathbb T}} L(\mathcal Q,(f_t,t\in \mathbb T) )\\
&\geq \inf_{\mathcal Q \in \mathcal M(\mathcal S)} \sup_{\mathcal P_{\mathbb T}} L(\mathcal Q,(\mathbf1_{A_t},t\in \mathbb T) ),
\end{align*}
and it remains to show that
$$
\sup_{\mathcal P_{\mathbb T}} \inf_{\mathcal Q \in \mathcal M(\mathcal S)} L(\mathcal Q,(\mathbf 1_{A_t},t\in \mathbb T) )  \geq \sup_{ \mathcal P^\infty_{\mathbb T}} \inf_{\mathcal Q \in \mathcal M(\mathcal S)} L(\mathcal Q,(f_t,t\in \mathbb T) ) ,
$$
or, in view of Lemma \ref{saddle.lm},
$$
\sup_{\mathcal P_{\mathbb T}} \inf_{\mathcal Q \in \mathcal M(\mathcal S)} L(\mathcal Q,(\mathbf 1_{A_t},t\in \mathbb T) )  \geq  \inf_{\mathcal Q \in \mathcal M(\mathcal S)} L(\mathcal Q,(f^*_t,t\in \mathbb T) ).
$$
Fix a sequence $\{\varepsilon_n\}_{n\geq 1}$ of positive numbers converging to zero. For each $n$, by Assumption \ref{cont.ass}, there exists $(A^n_t,t\in \mathbb T)\in \mathcal P_{\mathbb T}$ such that 
$$
\sum_{t\in \mathbb T} \int_{\mathcal S} d\mathcal Q(\theta) \mathbb E^{\theta} [Y_t \mathbf 1_{A^n_t}] + \varepsilon_n \geq \sum_{t\in \mathbb T}\int_{\mathcal S} d\mathcal Q(\theta) \mathbb E^{\theta} [Y_t f^*_t],\quad \forall \mathcal Q \in \mathcal M(\mathcal S).
$$
Indeed, let $\rho(h)$ be the functional in Assumption \ref{cont.ass}, and let $d_n$ be such that $2r\rho(d_n)\leq \varepsilon_n$ (where we recall that $r$ is the size of the time grid). 
As a compact metric space, $\mathcal S$ admits a finite cover by balls of radius $d_n$. Let $\mathcal S^n = \{\theta_k,k=1,\dots,N_n\}$ be the centers of such balls and let $(A^n_t,t\in \mathbb T)$ be given by Lemma \ref{thin.lm} with the set of scenarios $\mathcal S^n$. For any $\theta\in \mathcal S$, denote by $\psi_n(\theta)$ the element of $\mathcal S^n$ which is closest to $\theta$. 
Then, 
\begin{align*}
&\left|\sum_{t\in \mathbb T} \int_{\mathcal S} d\mathcal Q(\theta) \mathbb E^{\theta} [Y_t \mathbf 1_{A^n_t}] - \sum_{t\in \mathbb T}\int_{\mathcal S} d\mathcal Q(\theta) \mathbb E^{\theta} [Y_t f^*_t]\right|\\
& \leq \left|\sum_{t\in \mathbb T} \int_{\mathcal S} d\mathcal Q(\theta) \mathbb E [D^\theta Y_t (\mathbf 1_{A^n_t} - f^*_t)] \right|\\
& = \left|\sum_{t\in \mathbb T} \int_{\mathcal S} d\mathcal Q(\theta) \mathbb E [(D^\theta-D^{\psi_n(\theta)}) Y_t (\mathbf 1_{A^n_t} - f^*_t)] \right|\\
&\leq 2 \sum_{t\in \mathbb T} \int_{\mathcal S} d\mathcal Q(\theta) \mathbb E [|D^\theta-D^{\psi_n(\theta)}| Y_t ]\leq 2 r \rho(d_n)\leq \varepsilon_n.
\end{align*}

Finally, we can find a sequence $\{\mathcal Q_n\}$, which, by compactness, converges to $\overline{\mathcal Q}$, such that, 
\begin{align*}
&\liminf_n \inf_{\mathcal Q \in \mathcal M(\mathcal S)} R\left(\mathcal Q, \sum_{t\in \mathbb T} \int_{\mathcal S} d\mathcal Q(\theta) \mathbb E^{\theta} [Y_t \mathbf 1_{A^n_t}] \right) \\
& \geq  \liminf_n \inf_{\mathcal Q \in \mathcal M(\mathcal S)} R\left(\mathcal Q, \sum_{t\in \mathbb T} \int_{\mathcal S} d\mathcal Q(\theta) \mathbb E^{\theta} [Y_t f^*_t] - \varepsilon_n \right)\\
&\geq \liminf_n \left\{ R\left(\mathcal Q^n, \sum_{t\in \mathbb T} \int_{\mathcal S} d\mathcal Q^n(\theta) \mathbb E^{\theta} [Y_t f^*_t] - \varepsilon_n \right)-\varepsilon_n\right\}\\
& = \liminf_n  R\left(\mathcal Q^n, \sum_{t\in \mathbb T} \int_{\mathcal S} d\mathcal Q^n(\theta) \mathbb E^{\theta} [Y_t f^*_t] - \varepsilon_n \right).
\end{align*}
Because $R$ is l.s.c.~in the first argument and continuous and nondecreasing in the second argument, and because $\theta\mapsto \mathbb E^\theta[Y_t f_t]$ is continuous by Assumption \ref{cont.ass} 
\begin{align*}
&\liminf_n  R\left(\mathcal Q^n, \sum_{t\in \mathbb T} \int_{\mathcal S} d\mathcal Q^n(\theta) \mathbb E^{\theta} [Y_t f^*_t] - \varepsilon_n \right)\\
&\geq
R\left(\overline{\mathcal Q}, \liminf_n \sum_{t\in \mathbb T} \int_{\mathcal S} d\mathcal Q^n(\theta) \mathbb E^{\theta} [Y_t f^*_t] \right) \\&= R\left(\overline{\mathcal Q}, \sum_{t\in \mathbb T} \int_{\mathcal S} d\overline{\mathcal Q}(\theta) \mathbb E^{\theta} [Y_t f^*_t] \right)\geq \inf_{\mathcal Q \in \mathcal M(\mathcal S)} R\left({\mathcal Q}, \sum_{t\in \mathbb T} \int_{\mathcal S} d{\mathcal Q}(\theta) \mathbb E^{\theta} [Y_t f^*_t] \right)
\end{align*}
\end{proof}
\subsection{Proof of Theorem \ref{thm:minmax}}
\label{sec:proofminimax}
In this section we prove Theorem \ref{thm:minmax} by extending Theorem \ref{thm:infscenarios}. 
Clearly, 
$$
\sup_{\tau \in \T} \inf_{\mathcal Q \in \mathcal M(\mathcal S)} G\left(\tau,\mathcal Q\right)\leq  \inf_{\mathcal Q \in \mathcal M(\mathcal S)}\sup_{\tau \in \T} G\left(\tau,\mathcal Q\right),
$$
the difficulty is to prove the opposite inequality. Let $n$ be a positive integer, and denote $\mathbb T_n = \{kT/n,k=1,\dots,n\}$ and $\tau_n[\tau] = \min\{t\in \mathbb T_n, t\geq \tau\}$. Then, by Theorem \ref{thm:infscenarios},
\begin{align*}
\sup_{\tau \in \T} \inf_{\mathcal Q \in \mathcal M(\mathcal S)} G\left(\tau,\mathcal Q\right) &\geq \sup_{\tau \in \mathbb T_n} \inf_{\mathcal Q \in \mathcal M(\mathcal S)} G\left(\tau,\mathcal Q\right) =  \inf_{\mathcal Q \in \mathcal M(\mathcal S)} \sup_{\tau \in \mathbb T_n}G\left(\tau,\mathcal Q\right).
\end{align*}
Fix a sequence $\{\varepsilon_n\}_{n\geq 1}$ of positive numbers converging to zero. For each $n$, there exists $\mathcal Q^n \in \mathcal M (\mathcal S)$ such that 
$$
\inf_{\mathcal Q \in \mathcal M(\mathcal S)} \sup_{\tau \in \mathbb T_n}G\left(\tau,\mathcal Q\right) \geq \sup_{\tau \in \mathbb T_n}G\left(\tau,\mathcal Q^n\right)-\varepsilon_n.
$$
By compactness of $\mathcal M (\mathcal S)$, there exists $\overline{\mathcal Q}$ such that $\mathcal Q^n \to \overline{\mathcal Q}$. Then, since $R$ is nondecreasing in the second argument, and in view of Assumption \ref{continuity},
\begin{align}
\liminf_m\inf_{\mathcal Q \in \mathcal M(\mathcal S)} \sup_{\tau \in \mathbb T_m}G\left(\tau,\mathcal Q\right)& \geq \liminf_{m}\{ \sup_{\tau \in \mathbb T_m}R\left(\mathcal Q^m, \int_{\mathcal S} d \mathcal Q^m(\theta) \mathbb E^\theta[Y_{\tau}]\right) - \varepsilon_m\}\notag\\
& = \liminf_{m} R\left(\mathcal Q^m, \sup_{\tau \in \mathbb T_m}\int_{\mathcal S} d \mathcal Q^m(\theta) \mathbb E^\theta[Y_{\tau}]\right) \notag\\
&\geq R\left(\overline{\mathcal Q}, \liminf_{m} \sup_{\tau \in \mathbb T_m}\int_{\mathcal S} d \mathcal Q^m(\theta) \mathbb E^\theta[Y_{\tau}]\right) .\label{estimR}
\end{align}
Let us now fix a constant $\varepsilon>0$ and choose $\bar\tau \in \mathcal T$, such that 
$$
\sup_{\tau\in \mathcal T} \int_{\mathcal S} d\overline{\mathcal Q}(\theta) \mathbb E^\theta[Y_\tau]\leq \int_{\mathcal S} d\overline{\mathcal Q}(\theta) \mathbb E^\theta[Y_{\bar\tau}]+\varepsilon. 
$$
We can decompose: 
\begin{align*}
\sup_{\tau \in \mathbb T_m}\int_{\mathcal S} d \mathcal Q^m(\theta) \mathbb E^\theta[Y_{\tau}] &= \int_{\mathcal S} d \overline{\mathcal Q}(\theta) \mathbb E^\theta[Y_{\bar\tau}]\\
 & + \int_{\mathcal S} d {\mathcal Q}^m(\theta) \mathbb E^\theta[Y_{\bar\tau}] -\int_{\mathcal S} d \overline{\mathcal Q}(\theta) \mathbb E^\theta[Y_{\bar\tau}] \\
& + \sup_{\tau \in \mathbb T_m}\int_{\mathcal S} d \mathcal Q^m(\theta) \mathbb E^\theta[Y_{\tau}]-\int_{\mathcal S} d {\mathcal Q}^m(\theta) \mathbb E^\theta[Y_{\bar\tau}].
\end{align*}
By Assumption \ref{cont.ass}, $\theta\mapsto \mathbb E^{\theta}[Y_{\bar \tau}]$ is continuous, and therefore, for $m$ sufficiently large,
$$
\left|\int_{\mathcal S} d {\mathcal Q}^m(\theta) \mathbb E^\theta[Y_{\bar\tau}] -\int_{\mathcal S} d \overline{\mathcal Q}(\theta) \mathbb E^\theta[Y_{\bar\tau}] \right|\leq \varepsilon. 
$$
On the other hand, 
\begin{align*}
\sup_{\tau \in \mathbb T_m}\int_{\mathcal S} d \mathcal Q^m(\theta) \mathbb E^\theta[Y_{\tau}]-\int_{\mathcal S} d {\mathcal Q}^m(\theta) \mathbb E^\theta[Y_{\bar\tau}]\geq \int_{\mathcal S} d {\mathcal Q}^m(\theta) \mathbb E^\theta[Y_{\tau_m[\bar\tau]}-Y_{\bar\tau}],
\end{align*}
where the last term is $o(1)$ as $m\to \infty$ by Assumption \ref{holder.ass}.  Finally, for $m$ large enough
$$
\sup_{\tau \in \mathbb T_m}\int_{\mathcal S} d \mathcal Q^m(\theta) \mathbb E^\theta[Y_{\tau}] \geq \sup_{\tau \in \mathcal T}\int_{\mathcal S} d \overline{\mathcal Q}(\theta) \mathbb E^\theta[Y_{\tau}] -3\varepsilon,
$$
so that 
$$
\liminf_m \sup_{\tau \in \mathbb T_m}\int_{\mathcal S} d \mathcal Q^m(\theta) \mathbb E^\theta[Y_{\tau}] \geq \sup_{\tau \in \mathcal T}\int_{\mathcal S} d \overline{\mathcal Q}(\theta) \mathbb E^\theta[Y_{\tau}] -3\varepsilon,
$$
and since $\varepsilon$ was arbitrary, we also have
$$
\sup_{\tau \in \mathbb T_m}\int_{\mathcal S} d \mathcal Q^m(\theta) \mathbb E^\theta[Y_{\tau}] \geq \sup_{\tau \in \mathcal T}\int_{\mathcal S} d \overline{\mathcal Q}(\theta) \mathbb E^\theta[Y_{\tau}].
$$
Plugging this estimate into \eqref{estimR} yields
$$
\liminf_n\inf_{\mathcal Q \in \mathcal M(\mathcal S)} \sup_{\tau \in \mathbb T_n}G\left(\tau,\mathcal Q\right) \geq R\left(\overline{\mathcal Q}, \sup_{\tau \in \mathcal T}\int_{\mathcal S} d \overline{\mathcal Q}(\theta) \mathbb E^\theta[Y_{\tau}]\right) = \sup_{\tau \in \mathcal T} G(\tau,\overline {\mathcal Q}).
$$
Finally, this implies
$$
\liminf_n\inf_{\mathcal Q \in \mathcal M(\mathcal S)} \sup_{\tau \in \mathbb T_n}G\left(\tau,\mathcal Q\right) \geq\inf_{\mathcal Q \in \mathcal M(\mathcal S)}\sup_{\tau \in \mathcal T} G(\tau,\blue{\mathcal Q}),
$$
which finishes the proof. 

\bibliographystyle{chicago}
\bibliography{references}
\end{document}